\documentclass[a4paper]{amsproc}
\usepackage{amssymb}
\usepackage{amscd}
\usepackage[dvips]{graphicx}

\theoremstyle{plain}
 \newtheorem{thm}{Theorem}[section]
 \newtheorem{prop}{Proposition}[section]
 \newtheorem{lem}{Lemma}[section]
 \newtheorem{cor}{Corollary}[section]
\theoremstyle{definition}
 
 \newtheorem{rem}{Remark}[section]

\numberwithin{equation}{section}

\renewcommand{\setminus}{\smallsetminus}

\def\ji {\char'032}
\def\ja {\char'037}
\def\m  {\char'176}

 \font\srrm=wncyr8

\newcommand{\T}{\mathbb{T}}

\newcommand{\R}{\mathbb{R}}

\newcommand{\C}{\mathbb{C}}
\newcommand{\CP}{\mathbb{P}}

\DeclareMathOperator{\diag}{\mathrm{diag}}

\title[The Jacobi-Rosochatius problem on an ellipsoid]{The Jacobi-Rosochatius problem on an ellipsoid: the Lax representations and billiards}
\author[Jovanovi\'c]{\sc Bo\v zidar Jovanovi\'c}

\address{
Mathematical Institute \\
Serbian Academy of Sciences and Arts \\
Kneza Mihaila 36, 11000 Belgrade\\
Serbia}
\email{bozaj@mi.sanu.ac.rs}

\begin{document}


\begin{abstract}
The Lax representations of the geodesic flow, the
Jacobi-Roso\-chatius problem and its perturbations by means of
separable polynomial potentials, on an ellipsoid are constructed.
We prove complete integrability in the case of a generic symmetric
ellipsoid and describe analogous systems on complex projective
spaces. Also, we consider billiards within an ellipsoid under the
influence of the Hook and Rosochatius potentials between the
impacts. A geometric interpretation of the integrability analogous
to the classical Chasles and Poncelet theorems is given.
\end{abstract}

\maketitle

\tableofcontents

\section{Introduction}

A well known Jacobi's problem describes the motion of a material
point on a $n$-dimensional ellipsoid $E^n$
\begin{equation}\label{ellipsoid}
E^n=\{x\in\R^{n+1}\,\vert\, \langle A^{-1} x,x\rangle=1\}
\end{equation}
under the influence of the Hook elastic force $-\sigma x$. Here $
A=\diag(a_0,a_1,\dots,a_n)$ is a positive definite matrix and
$\sigma$ is a real parameter \cite{J, Moser}. The Lagrangian of
the system is $ L(x,\dot x)=\frac12\langle \dot x,\dot
x\rangle-\frac{\sigma}{2}\langle x,x\rangle.$ The corresponding
Euler-Lagrange equations are
\begin{equation}\label{EL}
\ddot x=\lambda A^{-1}x-\sigma x,
\end{equation}
where the Lagrange multiplier is $\lambda=-({\langle A^{-1}\dot
x,\dot x\rangle-\sigma})/{\langle A^{-2}x,x\rangle}. $

By introducing the momenta $ y={\partial L}/{\partial \dot x}=\dot
x$, we can write \eqref{EL} in the Hamiltonian form
\begin{equation}\label{H}
\dot x=y, \qquad \dot y=-\frac{\langle
A^{-1}y,y\rangle-\sigma}{\langle A^{-2}x,x\rangle} A^{-1}x-\sigma
x
\end{equation}
on the cotangent bundle $T^*E^n$ realized by the constraints
\begin{equation}\label{E-constraints}
F_1=\langle A^{-1} x,x\rangle-1=0, \qquad F_2=\langle A^{-1}
x,y\rangle=0.
\end{equation}

The equations \eqref{H} are Hamiltonian with respect to the
Hamiltonian
\begin{equation}
H(x,y)=\langle y,\dot x\rangle - L(x,\dot x)\vert_{\dot
x=y}=\frac12\langle y,y\rangle +\frac{\sigma}{2}\langle x,x\rangle
\end{equation}
and the Dirac-Poisson bracket defined by
\begin{equation}
\{f_1,f_2\}_D=\{f_1,f_2\} -\frac{
\{F_1,f_1\}\{F_2,f_2\}-\{F_2,f_1\}\{F_1,f_2\}}{\{F_1,F_2\}},
\label{Dirac_bracket}
\end{equation}
where $\{\cdot,\cdot\}$ is the canonical Poisson bracket on
$\R^{2n+2}(x,y)$ (see \cite{AKN, Moser}). The Dirac-Poisson
bracket of the coordinate functions are:
\begin{equation*}
\{x_i,x_j\}_D=0, \quad\{y_i,y_j\}_D=-\frac{x_i y_j-x_j y_i}{a_i
a_j\langle A^{-2}x,x\rangle}, \quad
\{x_i,y_j\}_D=\delta_{ij}-\frac{x_i x_j}{a_i a_j\langle
A^{-2}x,x\rangle}.
\end{equation*}

Note that the Hamiltonian flow and the Dirac-Poisson structure is
defined not only on $T^*E^n$ but on the whole $\R^{2n+2}$ without
$x=0$.

If we set $\sigma=0$, the system  represents the geodesic flow on
the ellipsoid \eqref{ellipsoid}. The geodesic flow and the Jacobi
problem are one of the basic classical models of integrable
systems \cite{J, AKN}. For the ellipsoids with distinct semi-axes,
they are separable in the Jacobi elliptic coordinates
$(\lambda_0=0,\lambda_1,\dots,\lambda_n)$ defined as follows
\cite{AKN, Moser, Moser2, V}. Through every point $x\in \R^{n+1}$,
$x_0\cdot x_1\cdot\dots\cdot x_n\ne 0$, pass exactly $n+1$
mutually orthogonal confocal quadrics $\mathcal Q_{\lambda_0},
\mathcal Q_{\lambda_1}, \dots, \mathcal Q_{\lambda_n}$ ($\lambda_0
< a_0 < \lambda_1 <
 a_1 < \dots < \lambda_n < a_n$) given by the equation
\begin{equation}\label{confocal}
\mathcal Q_\lambda: \qquad \langle (A-\lambda)^{-1}
x,x\rangle=\sum_{i=0}^n\frac{x_i^2}{a_i-\lambda}=1.
\end{equation}
The original coordinates are defined up to a sign:
$$
x_k^2=\frac{\prod_i (a_k-\lambda_i)}{\prod_{i,i\ne k}(a_k-a_i)},
\quad k=0,\dots,n.
$$

In his celebrated paper \cite{Moser}, Moser found the Lax
representation $\dot L=[L,B]$,
$$
L=\left(\mathbf{I}_{n+1}-\frac{y\otimes y}{\langle
y,y\rangle}\right)\left(A-x\otimes x\right)
\left(\mathbf{I}_{n+1}-\frac{y\otimes y}{\langle
y,y\rangle}\right), \quad B=A^{-1} x \wedge A^{-1} y,
$$
of the associated system
\begin{equation}\label{ME}
\dot x=\frac{\partial \Phi}{\partial y}, \quad \dot
y=-\frac{\partial\Phi}{\partial x}, \quad \Phi=-\langle
y,A^{-1}y\rangle(1-\langle x,A^{-1}x\rangle)-\langle
x,A^{-1}y\rangle^2.
\end{equation}

The set $\Phi=0$ describes the set of tangents of the ellipsoid
$E^n=\mathcal Q_0$. Restricted to $\Phi=0$, the system \eqref{ME}
has the following important property: the moving line
$p(t)=\{x(t)+s y(t)\,\vert\,s\in\R\}$ has a point of contact $\xi$
with $E^{n}$ along nonparameterized geodesic $\xi(\tau)$.
Moreover, the eigenvalues $\eta_1,\dots,\eta_n$ of $L$ (different
from zero), define $n$ quadrics $\mathcal
Q_{\eta_1},\dots,\mathcal Q_{\eta_n}$ from the confocal family
\eqref{confocal}, such that $p(t)$ is simultaneously tangent to
$\mathcal Q_{\eta_1},\dots,\mathcal Q_{\eta_n}$ (a variant of the
classical {Chasles's} theorem \cite{Chasles}).

Although Moser's method  is applied in the construction of various
integrable models admitting Lax representations \cite{AHP}, Lax
representations with a spectral parameter of the geodesic flow and
the Jacobi problem on $E^n$, according to author's knowledge, are
not given yet.

In Section 2 we define double Jacobi flows and complex double
Jacobi flows and construct two different Lax representations, a
"small" $2\times 2$ one (Theorem \ref{T1}) and a "big"
$(n+1)\times (n+1)$ one (Theorem \ref{T2}). An appropriate
restrictions of the double Jacobi flows and the complex double
Jacobi flows are the Jacobi problems on the ellipsoids
\eqref{ellipsoid} and
\begin{equation}\label{CE}
E^{2n+1}=\{z\in\C^{n+1}\,\vert\, \langle A^{-1}z,\bar
z\rangle=1\},
\end{equation}
respectively, leading to the Lax representations  of the Jacobi
problem (Theorem \ref{T3}, Section 3). Note that a small $2\times
2$ Lax representation has the usual Lax matrix form of the {\it
Jacobi-Mamford} system (see \cite{Mum, Va}).

The Jacobi problem on the ellipsoid \eqref{CE} is invariant with
respect to the standard $\T^{n+1}$-action on
$\C^{n+1}$.\footnote{We use the complex notation to simplify the
reduction procedure. Equivalently one can consider the real space
$\R^{2n+2}$ and the ellipsoid $E^{2n+1}$ defined by the matrix
$A=(a_0,a_0,a_1,a_1,\dots,a_n,a_n)$.} It is well known (e.g., see
\cite{Moser, K}), that the reduced flow can be naturally
considered as a system describing the motion of a material point
on the ellipsoid \eqref{ellipsoid} under the influence of the Hook
and the Rosochatius potentials \cite{R}
\begin{equation}\label{Rpot}
V(x)=\frac{\sigma}{2}\langle
x,x\rangle+\frac12\sum_{k=0}^n\frac{\mu_k^2}{x_k^2}
\end{equation}
(Section 4). The Lax representation is also invariant with respect
to the standard $\T^{n+1}$-action and induce a small $2\times
2$-Lax representation with a spectral parameter for the
Jacobi-Rosochatius system on the ellipsoid $E^n$ (Theorem
\ref{T6}). Note that if instead of the $\T^{n+1}$-reduction, we
perform a $S^1$-reduction, with $S^1$ diagonally embedded in
$\T^{n+1}$, we obtain a natural mechanical system on a complex
projective space (Proposition \ref{Hermitian}), providing a class
of examples of Hermite-Liouville manifolds (see Remark \ref{H-L})
\cite{IK, T}.

If all semi-axes of the ellipsoid \eqref{ellipsoid} are distinct,
the Jacobi-Rosochatius system is separable in elliptic coordinates
and has $n$ independent commuting integrals which are quadratic in
momenta. The geodesic flows on symmetric ellipsoids are studied in
details in \cite{Kris}. In general, the geodesic flows are
integrable in a noncommutative sense. In Section 5, we prove
complete integrability of the symmetric Jacobi-Rosochatius system
both in a noncommutative (Theorem \ref{integrabilnost}) and the
Liouville sense (Theorem \ref{integrabilnost2}).

In Section 6 we consider the billiard system within an ellipsoid
under the influence of the potential \eqref{Rpot} between the
impacts. By a slight modification of the results given by Fedorov
\cite{Fe2}, we describe the billiard mapping and the Lax
representation (Theorem \ref{C3}). A geometric interpretation of
the integrability analogous to classical Chasles and Poncelet
theorems is given (Theorem \ref{C4}).

Finally, in the last section, by using the $2\times 2$ Lax
representation, we give the hierarchy of the Lax representations
for the separable potential perturbations of the
Jacobi-Rosochatius system on the ellipsoid $E^n$.

\section{Double Jacobi flows}

Together with the Jacobi problem, let us consider a system defined
by the Lagrangian $L(x,\xi,\dot x,\dot\xi)=\langle \dot
x,\dot\xi\rangle-\sigma\langle x,\xi\rangle$ and constrained on
the hypersurface
$$
\Sigma=\{(x,\xi)\in\R^{2n+2}\, \vert\, \langle x,
A^{-1}\xi\rangle-1=0\}.
$$

The Euler-Lagrange equations are
\begin{eqnarray*}
&& \ddot x=\lambda A^{-1}x-\sigma x,\\
&& \ddot \xi=\lambda A^{-1}\xi-\sigma \xi,
\end{eqnarray*}
where the Lagrange multiplier is
\begin{equation*}
\lambda=-\frac{\langle A^{-1}\dot x,\dot\xi\rangle-\sigma}{\langle
A^{-2}x,\xi\rangle}.
\end{equation*}

As in the case of the Jacobi problem, we can write the
corresponding Hamiltonian equations on $T^*\Sigma$. Canonically
conjugate momenta to $(x,\xi)$ are  $\eta={\partial L}/{\partial
\dot x}=\dot\xi$, $y={\partial L}/{\partial \dot \xi}=\dot x$. The
Hamiltonian function of the system is given by the Legendre
transformation:
\begin{equation*}
H(x,\xi,\eta,y)=\langle \dot x, \eta\rangle+\langle
\dot\xi,y\rangle-L=\langle y,\eta\rangle+\sigma\langle
x,\xi\rangle,
\end{equation*}
and the Hamiltonian equations are
\begin{eqnarray}
&& \label{DF1}\dot x=y, \qquad \dot y=-\frac{\langle A^{-1}y,\eta\rangle-\sigma}{\langle A^{-2}x,\xi\rangle} A^{-1}x-\sigma x,\\
&& \label{DF2}\dot \xi=\eta, \qquad \dot \eta=-\frac{\langle
A^{-1} y,\eta\rangle-\sigma}{\langle A^{-2}x,\xi\rangle}
A^{-1}\xi-\sigma \xi.
\end{eqnarray}
The symplectic structure on
\begin{equation}
\label{Sigma}T^*\Sigma:\quad G_1=\langle x, A^{-1}\xi\rangle-1=0,
\quad G_2= \langle y, A^{-1} \xi\rangle+\langle x,
A^{-1}\eta\rangle=0
\end{equation}
is the restriction of the standard symplectic structure on
$\R^{4n+4}(x,\xi,\eta,y)$. The corresponding Dirac-Poisson
structure is given by \eqref{Dirac_bracket}, where we should
replace the constraint functions $F_1, F_2$ by $G_1, G_2$ and
$\{\cdot,\cdot\}$ is the canonical Poisson structure on
$\R^{4n+4}(x,\xi,\eta,y)$.

By the analogy with the double Neumann system (see \cite{Su}), we
refer to \eqref{DF1}, \eqref{DF2} as a {\it double Jacobi flow}.

We can also consider the complexified
 phase space, an affine $2n+2$-dimensional variety in $\C^{4n+4}(x,\xi,\eta,y)$
 defined by \eqref{Sigma}.
Then we refer to \eqref{DF1}, \eqref{DF2} as a {\it complex double
Jacobi flow} (here, the time parameter $t$ can be both: real or
complex).

Let $A_\lambda=\diag(\lambda-a_0,\lambda-a_1,\dots,\lambda-a_n)$
and
$$
q_\lambda(x,\xi)=\langle A^{-1}_\lambda
x,\xi\rangle=\sum_{i=0}^n\frac{x_i\xi_i}{\lambda-a_i}.
$$

\begin{thm}\label{T1}
The (complex) double Jacobi flow \eqref{DF1}, \eqref{DF2} implies
the matrix equation
\begin{equation}
\dot{\mathcal L}(\lambda)=[\mathcal L(\lambda),\mathcal
A(\lambda)]
\end{equation}
with $2\times 2$ matrices depending on the parameter $\lambda$
\begin{eqnarray*}
&& \mathcal L(\lambda)=
\begin{pmatrix}
q_\lambda(x,\eta) & q_\lambda(y,\eta)+\sigma \\
-1-q_\lambda(x,\xi) & -q_\lambda(y,\xi)
\end{pmatrix},\\
&& \mathcal A(\lambda)=\frac{1}{\langle A^{-2}x,\xi\rangle}
\begin{pmatrix}
0 & \frac{1}{\lambda}(\sigma-\langle A^{-1}y,\eta\rangle)-\sigma\langle A^{-2} x,\xi\rangle\\
\langle A^{-2} x,\xi\rangle & 0
\end{pmatrix}.
\end{eqnarray*}
\end{thm}

\begin{proof} To simplify the calculations, consider the time
reparametrization
\begin{equation}\label{new-time}
dt=\langle A^{-2}x,\xi\rangle d\tau.
\end{equation}

The double Jacobi flow in the new time gets the symmetric form
\begin{eqnarray}
&&\nonumber \label{DF1*}x'=
\langle A^{-2}x,\xi\rangle y, \\
&&\nonumber y'=(\sigma-\langle A^{-1}y,\eta\rangle) A^{-1}x-\sigma \langle A^{-2}x,\xi\rangle x,\\
&& \label{DF2*}\xi'=\langle A^{-2}x,\xi\rangle\eta, \\
&&\nonumber\eta'=(\sigma-\langle A^{-1}
y,\eta\rangle)A^{-1}\xi-\sigma \langle A^{-2}x,\xi\rangle\xi,
\end{eqnarray}
where $(\cdot)'=\frac{d}{d\tau}(\cdot)=\langle
A^{-2}x,\xi\rangle\frac{d}{dt}(\cdot)$.

By using the constraints \eqref{Sigma} and the identity
$$
A^{-1}A^{-1}_\lambda=A^{-1}_\lambda
A^{-1}=(A^{-1}+A^{-1}_\lambda)/\lambda,
$$
we obtain the relations
\begin{eqnarray*}
&& (q_\lambda (x,\eta))'=(q_\lambda
(y,\xi))'=\frac{1}{\lambda}(1+\langle A^{-1}_\lambda
x,\xi\rangle)(\sigma-\langle
A^{-1}y,\eta\rangle)\\
&& \qquad\qquad\qquad   +\langle A^{-2}x,\xi\rangle (\langle y,
A^{-1}_\lambda\eta\rangle-\sigma\langle x,A^{-1}_\lambda \xi\rangle),\\
&& (q_\lambda(y,\eta))'=\frac{1}{\lambda}(\langle A^{-1}_\lambda
x,\eta\rangle+\langle A^{-1}_\lambda y,\xi\rangle)(\sigma-\langle
A^{-1}y,\eta\rangle)\\
&& \qquad\qquad\qquad\qquad-\sigma\langle A^{-2}
x,\xi\rangle(\langle x,A^{-1}_\lambda \eta\rangle+\langle
A^{-1}_\lambda y,\xi\rangle),\\
&& (q_\lambda(x,\xi))'=\langle A^{-2} x,\xi\rangle(\langle
x,A^{-1}_\lambda \eta\rangle+\langle A^{-1}_\lambda y,\xi\rangle),
\end{eqnarray*}

Now, the reader can simply verify $\mathcal L'=[\mathcal L,\langle
A^{-2} x,\xi\rangle\mathcal A]$.
\end{proof}

Therefore, the coefficients of the invariant polynomials of the
matrix $a(\lambda)\mathcal L(\lambda)$ are the first integrals of
the system. If all of $a_i$ are distinct, the integrals can be
written in the form
\begin{eqnarray*}
&& f_i=y_i\eta_i+\sigma x_i\xi_i+\sum_{j\ne
i}\frac{(y_ix_j-y_jx_i)(\eta_i\xi_j-\eta_j\xi_i)}{a_i-a_j},\label{f_i}\\
&& g_i=y_i\xi_i-x_i\eta_i, \qquad i=0,1,\dots,n.\label{g_i}
\end{eqnarray*}

\begin{thm} \label{T2}
The (complex) double Jacobi flow \eqref{DF1}, \eqref{DF2}
restricted to the invariant variety
\begin{equation}\label{invariant}
\langle A^{-1} x,\eta\rangle=\langle A^{-1} y,\xi\rangle=0.
\end{equation}
implies the matrix equation
\begin{equation}
\dot{\mathcal L}^*(\lambda)=[\mathcal A^*(\lambda),\mathcal
L^*(\lambda)]
\end{equation}
with $(n+1)\times (n+1)$ matrices depending on the parameter
$\lambda$
\begin{eqnarray*}
&& \mathcal L^*(\lambda)=\lambda(y\otimes \xi-x\otimes\eta)+y\otimes \eta+\sigma x\otimes\xi-\sigma A-\lambda^2A,\\
&& \mathcal A^*(\lambda)=\frac{1}{\langle
A^{-2}x,\xi\rangle}\left( A^{-1}y\otimes A^{-1}\xi-A^{-1}x\otimes
A^{-1}\eta+\lambda A^{-1} \right).
\end{eqnarray*}
\end{thm}

\begin{proof}From \eqref{DF2*} we get
\begin{eqnarray*}
&& (x\otimes \xi)'=\langle A^{-2}
x,\xi\rangle(y\otimes\xi+x\otimes\eta),\\
&& (y\otimes \eta)'=(\sigma-\langle A^{-1} y,\eta\rangle)(A^{-1}x
\otimes \eta+y\otimes A^{-1}\xi)-\sigma\langle
A^{-2}x,\xi\rangle(x\otimes\eta+y\otimes \xi),\\
&& (y\otimes\xi-x\otimes\eta)'=(\sigma-\langle A^{-1}
y,\eta\rangle)( A^{-1}x\otimes\xi-x\otimes A^{-1}\xi).\\
\end{eqnarray*}

Thus,
$$
(\mathcal L^*(\lambda))'=(\sigma-\langle A^{-1}
y,\eta\rangle)\left(\lambda(A^{-1}x \otimes \xi-x\otimes
A^{-1}\xi)+ A^{-1}x \otimes \eta+y\otimes A^{-1}\xi\right).
$$

Further, let us denote
$$
\Omega=A^{-1}y\otimes A^{-1}\xi-A^{-1}x\otimes A^{-1}\eta.
$$
Then, by using the constraints \eqref{Sigma}, we get
\begin{eqnarray*}
&&[\Omega,y\otimes\eta]=\langle A^{-1} y,\xi\rangle
A^{-1}y\otimes\eta+\langle A^{-1}x,\eta\rangle y\otimes
A^{-1}\eta\\
&&\qquad\qquad\qquad-\langle A^{-1}\eta,y\rangle
(A^{-1}x\otimes\eta+y\otimes A^{-1}\xi),\\
&&[\Omega,x\otimes\xi]=A^{-1}y\otimes \xi-\langle
A^{-1}\eta,x\rangle
A^{-1}x\otimes\xi-\langle A^{-1} y,\xi\rangle x\otimes A^{-1}\xi+x\otimes A^{-1}\eta,\\
&&[\Omega,y\otimes\xi-x\otimes \eta]=y\otimes
A^{-1}\eta-A^{-1}y\otimes\eta+ \langle A^{-1}x,\eta\rangle
(A^{-1}x\otimes\eta-x\otimes A^{-1}\eta)\\
&&\qquad\quad+\langle A^{-1}y,\xi\rangle
(A^{-1}y\otimes\xi-y\otimes A^{-1}\xi)- \langle
A^{-1}y,\eta\rangle (A^{-1}x\otimes\xi-x\otimes A^{-1}\xi),\\
&&[\Omega,A]=A^{-1}y\otimes \xi-A^{-1}x\otimes \eta-y\otimes
A^{-1}\xi+x\otimes A^{-1}\eta,
\end{eqnarray*}
implying
\begin{eqnarray*}
&& [\langle A^{-2} x,\xi\rangle\mathcal A^*(\lambda),\mathcal
L^*]=[\Omega+\lambda A^{-1},\lambda(y\otimes
\xi-x\otimes\eta)+y\otimes \eta+\sigma x\otimes\xi-\sigma
A-\lambda^2A]\\
&& \quad=(\sigma-\langle A^{-1}
y,\eta\rangle)\left(\lambda(A^{-1}x \otimes \xi-x\otimes
A^{-1}\xi)+ A^{-1}x \otimes \eta+y\otimes A^{-1}\xi\right)\\
&&\quad +\langle A^{-1} y,\xi\rangle (A^{-1}y\otimes\eta-\sigma
x\otimes A^{-1}\xi)+\langle A^{-1}x,\eta\rangle (y\otimes
A^{-1}\eta-\sigma
A^{-1}x\otimes \xi) \\
 &&\quad+\lambda\left(\langle A^{-1}x,\eta\rangle(A^{-1}x
\otimes \eta-x\otimes A^{-1}\eta)+\langle
A^{-1}y,\xi\rangle(A^{-1}y \otimes \xi-y\otimes A^{-1}\xi)\right).
\end{eqnarray*}

It remains to note that \eqref{invariant} defines an invariant
manifold of the system. Indeed, we have
$$
(\langle A^{-1}x,\eta\rangle)'=\langle A^{-2}x,\xi\rangle\langle
y,A^{-1}\eta\rangle+(\sigma-\langle A^{-1}y,\eta\rangle)\langle
A^{-2}x,\xi\rangle-\sigma\langle A^{-2}x,\xi\rangle=0.
$$
\end{proof}

Note that the matrix $\mathcal L(\lambda)$ is invariant under the
transformations of the phase space given by the
$(\R^*)^{n+1}$-action (respectively $({\C^*})^{n+1}$-action):
$$
(x_i,\xi_i,\eta_i,y_i) \longmapsto (s_i x_i, s_i^{-1}
\xi_i,s_i^{-1} \eta_i,s_i y_i), \qquad s_i\ne 0, \qquad
i=0,1,\dots,n,
$$
while the matrix $\mathcal L^*(\lambda)$ is $\R^*$-invariant
(respectively $\C^*$-invariant):
$$
(x_i,\xi_i,\eta_i,y_i) \longmapsto (s x_i, s^{-1} \xi_i,s^{-1}
\eta_i,s y_i), \qquad s\ne 0, \qquad i=0,1,\dots,n.
$$

\section{The Lax representations of the Jacobi problem}

The equations
$$
x=\xi, \qquad y=\eta
$$
define the invariant manifold of \eqref{DF1}, \eqref{DF2}, so the
double Jacobi flow contains as a subsystem the Jacobi problem
\eqref{H}. Also, for $x=\xi$, $y=\eta$, the condition
\eqref{invariant} is satisfied and the above theorems imply Lax
representations for the Jacobi problem. In particular, when we set
$\sigma=0$, we get the Lax representations for the geodesic flow
on a ellipsoid.

\begin{thm}\label{T3}
The  Jacobi problem \eqref{H} implies the matrix equations
\begin{equation}
\label{LA1} \dot{\mathcal L}(\lambda)=[\mathcal
L(\lambda),\mathcal A(\lambda)]
\end{equation}
and
\begin{equation}
 \label{LA2} \dot{\mathcal
L}^*(\lambda)=[\mathcal A^*(\lambda),\mathcal L^*(\lambda)]
\end{equation}
with $2\times 2$ and $(n+1)\times (n+1)$ matrices depending on the
parameter $\lambda$
\begin{eqnarray*}
&& \mathcal L(\lambda)=
\begin{pmatrix}
q_\lambda(x,y) & q_\lambda(y,y)+\sigma \\
-1-q_\lambda(x,x) & -q_\lambda(y,x)
\end{pmatrix},\\
&& \mathcal A(\lambda)=\frac{1}{\langle A^{-2}x,x\rangle}
\begin{pmatrix}
0 & \frac{1}{\lambda}(\sigma-\langle A^{-1}y,y\rangle)-\sigma\langle A^{-2} x,x\rangle\\
\langle A^{-2} x,x\rangle & 0
\end{pmatrix},\\
&& \mathcal L^*(\lambda)=\lambda(y\wedge x)+y\otimes y+\sigma x\otimes x-\sigma A-\lambda^2A,\\
&& \mathcal A^*(\lambda)=\frac{1}{\langle A^{-2}x,x\rangle}\left(
A^{-1}y\wedge A^{-1}x+\lambda A^{-1} \right),
\end{eqnarray*}
respectively.
\end{thm}

From the Lax representations we obtain the well known form of the
integrals of the Jacobi problem given by Moser \cite{Moser}. Note
that the $2\times 2$ Lax matrix $\mathcal L(\lambda)$ has the
usual Lax matrix form of the {\it Jacobi-Mamford} systems (see
\cite{Mum, Va}). Also note, if all $a_i$ are distinct, the Lax
representations \eqref{LA1} and \eqref{LA2}  are equivalent to the
Jacobi problem \eqref{H},  up to the action of groups $\mathbb
Z_2^{n+1}$, $\mathbb Z_2$  generated by reflections
\begin{eqnarray*}
&&(x_i,y_i)\longmapsto (s_ix_i,s_iy_i), \quad s_i=\pm 1, \qquad
i=0,1,\dots,n,\\
&& (x,y)\longmapsto s(x,y), \qquad\qquad s=\pm 1.
\end{eqnarray*}
The time reperemetrezation \eqref{new-time}, for $x=\xi$,
coincides with the time reparametrization
$dt=\lambda_1\dots\lambda_n d\tau$ used in the integration of the
geodesic flow \cite{Wa, Br, F3}.

\section{The Jacobi-Rosochatius problem on an ellipsoid}

Consider the complex double Jacobi flow \eqref{DF1}, \eqref{DF2}
on the invariant real variety
\begin{equation}\label{CPS}
x=z, \quad \xi=\bar z, \quad y=p, \quad \eta=\bar p.
\end{equation}

The equations \eqref{Sigma}, \eqref{CPS} define the cotangent
bundle $T^*E^{2n+1}$ of the $2n+1$-dimensional real ellipsoid in
the complex space \eqref{CE}
\begin{equation}\label{complex-ellipsoid}
T^*E^{2n+1}: \quad \langle A^{-1}z,\bar z\rangle-1=0, \quad
\langle A^{-1} z,\bar p \rangle+\langle A^{-1}\bar z,p\rangle=0.
\end{equation}

The complex double Jacobi flow restricted to $T^*E^{2n-1}$
\begin{equation}
\label{CP1}\dot z=p, \qquad \dot p=-\frac{\langle A^{-1}p,\bar
p\rangle-\sigma}{\langle A^{-2}z,\bar z\rangle} A^{-1}z-\sigma z,
\end{equation}
describes the motion of a material point on $E^{2n+1}$, under the
influence of the elastic force $-\sigma z$. From Theorem \ref{T1}
we see that the system \eqref{CP1} implies the matrix equation
\begin{equation}\label{LA3}
\dot{\mathcal L}(\lambda)=[\mathcal L(\lambda),\mathcal
A(\lambda)],
\end{equation}
where
\begin{eqnarray*}
 && \nonumber \mathcal L(\lambda)=\begin{pmatrix}
q_\lambda(z,\bar p) & q_\lambda(p,\bar p)+\sigma \\
-1-q_\lambda(z,\bar z) & -q_\lambda(p,\bar z)
\end{pmatrix},\label{CL}
\\
&& \nonumber \mathcal A(\lambda)=\frac{1}{\langle A^{-2}z,\bar
z\rangle}
\begin{pmatrix}
0 & \frac{1}{\lambda}(\sigma-\langle A^{-1}p,\bar p\rangle)-\sigma\langle A^{-2} z,\bar z\rangle\\
\langle A^{-2} z,\bar z\rangle & 0
\end{pmatrix}.\label{CA}
\end{eqnarray*}

The system \eqref{CP1}, as well as the Lax representation
\eqref{LA3}, is invariant with respect to the Hamiltonian torus
action on $T^*E^{2n+1}$:
\begin{equation}\label{torus}
(z_k,p_k)\longmapsto e^{i\varphi_k}(z_k,p_k), \qquad
k=0,1,\dots,n,
\end{equation}
with the momentum mapping
\begin{equation}\label{actions}
\mathbf{h}=(h_0,h_1,\dots,h_n), \qquad h_k=-\frac{i}2
g_k=\frac{i}2(z_k\bar p_k-p_k\bar z_k).
\end{equation}

\subsection{The Jacobi system on a complex projective space}
In particular, the Hamiltonian $h=h_0+\dots+h_n$ induces a
$S^1$-action
\begin{equation}\label{circle}
(z,p)\mapsto e^{i\varphi}(z,p).
\end{equation}

The symplectic reduced space $h^{-1}(0)/S^1$ is simplectomorphic
to the cotangent bundle of the complex projective space $\CP^n$.
The reduced system is a natural mechanical system with the kinetic
energy determined by the ellipsoidal metric on $\CP^n$, the
submersion metric with respect to $S^1$-action, under the
influence of the "elastic" force. From Theorems
\ref{integrabilnost} and \ref{integrabilnost2}, where we identify
$\C^{n+1}\cong \R^{2n+2}$ and set $
A=\diag(a_0,a_0,a_1,a_1,\dots,a_n,a_n)$, we have integrability of
the Jacobi problem \eqref{complex-ellipsoid} for an arbitrary
choice of $a_i$. Similarly, the reduced flow is integrable. A
general treatment of the integrability of the reduced systems is
given in \cite{Zung, JB}. Here, the orbit of the $S^1$-action are
tangent to the isotropic tori $\mathcal T$, which lay in
$h^{-1}(0)$, and $\mathcal T/S^1$ are invariant isotropic tori for
the reduced flow.

In the next statement we describe the reduced system on $\CP^n$.
Let $\pi: \C^{n+1}\setminus\{0\}\to \CP^n$ be the canonical
projection of $w=(w_0,\dots,w_n)\in \C^{n+1}\setminus\{0\}$ to
$[w]=[w_0:\dots:w_n]\in\CP^n$, with respect to the $C^*$-action.

\begin{prop} \label{Hermitian}
The reduced Jacobi problem on $\CP^n$ is a natural
mechanical system with the kinetic energy determined by the metric
$$
\tilde g_A(\pi_* X,\pi_* X)\vert_{[w]}=\frac{\langle w,A\bar
w\rangle\langle X,A\bar X\rangle-\langle X,A\bar w\rangle\langle
w,A\bar X\rangle}{\langle w,A\bar w\rangle\langle w,\bar
w\rangle}, \quad X\in T_w\C^{n+1}\setminus\{0\}
$$
under the influence of the potential field $V([w])=\sigma{\langle
w, A\bar w\rangle}/{2\langle w,\bar w\rangle}$.
\end{prop}

\begin{proof} Under the change of variables $z_i=\sqrt{a_i} w_i$,
the Jacobi problem transforms to the system on a sphere $S^{2n+1}:
\, \langle w,\bar w\rangle=1$ with the Lagrangian function
$$
L(w,\dot w)=\frac12\langle \dot w,A\dot{\bar w}\rangle
-\frac{\sigma}2\langle w,A\bar w\rangle.
$$

It is well known that the reduced system on $h^{-1}(0)/S^1 \cong
T^*\CP^n$ can be described by the $S^1$-Lagrange-Routh reduction
with a zero value of the $S^1$-momentum mapping (e.g., see
\cite{AKN}). Since it is convenient to work with homogeneous
coordinates, we firstly extend $L$ to the $C^*$-invariant
Lagrangian function
\begin{equation}\label{LE}
L=\frac{\langle \dot w,A\dot w\rangle}{2\langle w,\bar w\rangle}
-\frac{\sigma\langle w,A\bar w\rangle}{2\langle w,\bar w\rangle},
\end{equation}
defined on $\C^{n+1}\setminus\{0\}$ and then perform the
Lagrange-Routh reduction with respect to the $C^*$-action.

Obviously, the reduced potential is $V([w])=\sigma{\langle w,
A\bar w\rangle}/{2\langle w,\bar w\rangle}$. Further, the
Lagrangian \eqref{LE} defines the Riemannian and Hermitian metrics
$$
g_A(X,Y)=\frac{\langle X,A \bar Y\rangle+\langle \bar
X,AY\rangle}{2\langle w,\bar w\rangle}, \quad
h_A(X,Y)=\frac{\langle X,A \bar Y\rangle}{\langle w,\bar
w\rangle}, \quad  X, Y\in T_w\C^{n+1}\setminus\{0\},
$$
respectively. The reduced system is a natural mechanical system
$(\CP^n,\tilde g_A, V([w]))$, where $\tilde g_A$ is the submersion
metric with respect to the $C^*$-action \cite{AKN}.

At every $w \in \C^{n+1}\setminus\{0\}$ we have a decomposition
$$
T_w \C^{n+1}\setminus\{0\}=\C^{n+1}=\mathcal V_w \oplus \mathcal
H_w,
$$
where $\mathcal V_w$ is the tangent space of the orbit of
$C^*$-action through $w$ ({\it vertical space} at $w$) and
$\mathcal H_w$ is its $g_A$-orthogonal complement ({\it horizontal
space} at $w$). Since $\mathcal V_w$ is the complex line through
$w$, its real $g_A$-orthogonal complement coincides with its
$h_A$-Hermitian orthogonal complement
$$
\mathcal H_w=\{ X \,\vert\, h_A(X,w)=0\}.
$$

Let $X,Y\in T_w\C^{n+1}\setminus\{0\}$. By definition, the
submersion metric $\tilde g_A(\pi_* X,\pi_*)\vert_{[w]}$ on
$\CP^n$ by is equal to $g_A(X',Y')\vert_w$, where $X'$ and $Y'$
are horizontal components of $X$ and $Y$
$$
X'=X-\frac{h_A(X,w)}{h_A(w,w)}w, \quad
Y'=Y-\frac{h_A(Y,w)}{h_A(w,w)}w.
$$

Therefore
\begin{eqnarray*}
\tilde g_A(\pi_* X,\pi_* X)\vert_{[w]} &=&
g_A(X',X')\vert_w=h_A(X',X')\vert_{w}=h_A(X',X)\\
&=& h_A(X,X)-\frac{h_A(w,X)h_A(X,w)}{h_A(w,w)}\\
&=& \frac{\langle w,A\bar w\rangle\langle X,A\bar X\rangle-\langle
X,A\bar w\rangle\langle w,A\bar X\rangle}{\langle w,A\bar
w\rangle\langle w,\bar w\rangle}
\end{eqnarray*}
\end{proof}

\begin{rem}\label{H-L}
Note that $(\CP^n, \tilde g_A)$ is an example of a
Hermite-Liouville manifold (see \cite{IK, T}). In particular, if
$A$ is the identity matrix, then $\tilde g_{A}$ is the standard
Fubini-Study metric on $\CP^n$. The integrability of the geodesic
flow of the Fubini-Study metric is proved by Thimm \cite{Th} and
Boyer, Kalnins and Winternitz \cite{BKW}. Further, besides the
Hook potential in \eqref{CP1}, we can add other separable
polynomial potentials $V(z)$ (see the last section), and by the
Maupertuis principle (e.g., see \cite{AKN}),
$(\CP^n,(c-V([w]))\tilde g_A)$, $c>\max V([w])$ are examples of
Hermite-Liouville manifolds as well.
\end{rem}

\subsection{Reduction to the Jacobi-Rosochatius problem.}
Now we shall perform the reduction with respect to the torus
action \eqref{torus}. It is well known that we obtain a natural
mechanical system under the influence of the Rosochatius potential
\cite{R, Moser, K}.

Introduce the canonical polar coordinates
$(x_k,\varphi_k,y_k,h_k)$, where $h_k$ are given by
\eqref{actions} and $\varphi_k,y_k$ by
\begin{equation}\label{polar}
z_k=x_ke^{i\varphi_k}, \quad p_k=y_k
e^{i\varphi_k}+i\frac{h_k}{x_k} e^{i\varphi_k} , \qquad
k=0,1,\dots,n.
\end{equation}

The Hamiltonian of the Jacobi problem in new coordinates reads
$$
H=\frac12\langle p,\bar p \rangle+\frac{\sigma}{2}\langle z,\bar
z\rangle=\frac12\langle y,y\rangle+\frac{\sigma}{2}\langle
x,x\rangle+\frac12\sum_{k=0}^n\frac{h_k^2}{x_k^2}.
$$

We see that $\varphi_k$ are cyclic variables of the system.
Consider the level set of the momentum mapping
$$
(T^*E^{2n+1})_\mu: \qquad \mathbf h=(\mu_0,\mu_1,\dots,\mu_n)
$$
(some of constants $\mu_k$ may be equal to zero).

The equations of the system on $(T^*E^{2n+1})_\mu$ separate on
{\it reconstruction equations}
$$
x_k^2\dot\varphi_k=\mu_k, \qquad k=0,1,\dots,n.
$$
and the {\it reduced system} on $(T^*E^{2n+1})_\mu/\T^{n+1}$ in
variables $(x,y)$:
\begin{equation}\label{RH}
\dot x_k=y_k, \quad \dot y_k=-\frac{\langle
A^{-1}y,y\rangle+\langle A^{-1}
\frac{\mu}x,\frac{\mu}x\rangle-\sigma}{\langle A^{-2}x,x\rangle}
a_k^{-1}x_k-\sigma x_k+\frac{\mu_k^2}{x_k^3},
\end{equation}
where ${\mu}/x=(\mu_0/x_0,\mu_1/x_1,\dots,\mu_n/x_n)$.

Note that, according to the definition \eqref{polar}, we have
$x_k\ge 0$. Also, the variables $(x,y)$ satisfy the constraints
\eqref{E-constraints}. However, we can consider the system
\eqref{RH} on the whole cotangent bundle \eqref{E-constraints}.
Then it represents the motion of a material point on a ellipsoid
with the potential energy $V(x)$ having two terms: the Hook and
the Rosochatius potential \eqref{Rpot}.

Applying the change of variables \eqref{polar} to the Lax
representation \eqref{LA3} on the invariant set
$(T^*E^{2n+1})_\mu$, after subtracting the multiple of the
identity matrix from $\mathcal L(\lambda)$, we get the following
statement.

\begin{thm}\label{T6} Suppose that the eigenvalues $a_i$ of the matrix $A$
are distinct. Up to the action of the group $\mathbb Z_2^{n+1}$
generated by the reflections
$$
(x_i,y_i)\longmapsto (s_ix_i,s_iy_i), \quad s_i=\pm 1, \qquad
i=0,1,\dots,n,
$$
the Jacobi-Rosochatius problem \eqref{RH} is equivalent to the
matrix equation
\begin{equation}\label{LA4}
\dot{\mathcal L}(\lambda)=[\mathcal L(\lambda),\mathcal
A(\lambda)]
\end{equation}
with $2\times 2$ matrices depending on the parameter $\lambda$
\begin{eqnarray*}
&& \mathcal L(\lambda)=
\begin{pmatrix}
q_\lambda(x,y) & q_\lambda(y,y)+q_\lambda(\frac{\mu}{x},\frac{\mu}{x})+\sigma \\
-1-q_\lambda(x,x) & -q_\lambda(y,x)
\end{pmatrix},\\
&& \mathcal A(\lambda)=\frac{1}{\langle A^{-2}x,x\rangle}
\begin{pmatrix}
0 & \frac{1}{\lambda}(\sigma-\langle A^{-1}y,y\rangle-\langle A^{-1}\frac{\mu}x,\frac{\mu}x\rangle)-\sigma\langle A^{-2} x,x\rangle\\
\langle A^{-2} x,x\rangle & 0
\end{pmatrix}.
\end{eqnarray*}
\end{thm}

The integrals obtained by \eqref{LA4} can be written in the form
\begin{equation}\label{f-JR}
f_i=y_i^2+\sigma x_i^2+\frac{\mu_i^2}{x_i^2}+\sum_{j\ne
i}\frac{1}{a_i-a_j}\left((y_ix_j-y_jx_i)^2+\frac{\mu_i^2x^2_j}{x^2_i}+\frac{\mu^2_jx^2_i}{x^2_j}\right),
\end{equation}
$i=0,1,\dots,n$. This is a commutative set of functions, both with
respect to the canonical Poisson brackets $\{\cdot,\cdot\}$ and
the Dirac-Poisson bracket $\{\cdot,\cdot\}_D$.

Let us note that a possible alternative approach in the
construction of the Lax representations for the Jacobi-Rosochatius
problem \eqref{RH} is by using the Lax representations of the
Neumann problem (e.g., see \cite{Su}) and the well known
correspondence between the Neumann problem and the geodesic flow
on an ellipsoid \cite{Knorr1, K}. An algebro-geometric study of
the Neumann system on a sphere $S^n$ with the addition of the
Rosochatius potential is given in \cite{GHWH}.

\section{Symmetric ellipsoids. Complete integrability}

Consider the case of a symmetric ellipsoid:
\begin{eqnarray}
&& \nonumber a_i=\alpha_0,  \qquad i\in I_0=\{0,\dots,k_0-1\},\\
&& \label{symmetric} a_i=\alpha_1, \qquad i\in I_1=\{k_0,\dots,k_0+k_1-1\},\\
&& \nonumber \dots \\
&& \nonumber a_i=\alpha_r, \qquad i\in
I_r=\{k_0+\dots+k_{r-1},\dots,k_0+\dots+k_{r-1}+k_r-1\},
\end{eqnarray}
$\alpha_i\ne\alpha_j, i\ne j$, $k_0+\dots+k_r=n+1$. Then the
Jacobi-Rosochatius problem \eqref{RH} is not equivalent to the Lax
representation \eqref{LA4}, but still implies it. From the Lax
representation \eqref{LA4}, the invariants of $\mathcal
L(\lambda)$ are integrals of the flow. By using the relation
\begin{eqnarray}
\nonumber \det\mathcal L(\lambda)&=&
(1+q_{\lambda}(x,x))\left(q_\lambda(y,y)+q_\lambda\left(\frac{\mu}{x},
\frac{\mu}{x}\right)+\sigma\right)-q_\lambda(x,y)^2\\
&=&\sigma+\sum_{s=0}^r \frac{\tilde
f_s}{\lambda-\alpha_s}+\frac{P_s}{(\lambda -
\alpha_s)^2}+\sum_{i=0}^n\frac{\mu_i^2}{(\lambda-a_i)^2},
 \label{moser-rel*}
\end{eqnarray}
we get the integrals
\begin{eqnarray}
\label{f-JR*}
&& \tilde f_s=\sum_{i\in I_s} \left(y_i^2+\sigma
x_i^2+\frac{\mu_i^2}{x_i^2}+\sum_{j\notin
I_s}\frac{P_{ij}}{a_i-a_j}\right),\\
&& \label{P-s} P_s=\sum_{i,j\in I_s,i<j} P_{ij}, \qquad k_s=\vert
I_s\vert \ge 2,
\end{eqnarray}
where $P_{ij}$ are given by
\begin{equation*}
P_{ij}=(y_ix_j-x_iy_j)^2+\frac{\mu_i^2x^2_j}{x^2_i}+\frac{\mu^2_jx^2_i}{x^2_j},
\end{equation*}
while for $\vert I_s\vert=1$ we set $P_s \equiv 0$. Whence, we
have $\rho$ nontrivial integrals among $P_s$, where $\rho$ is the
number of sets $I_s$ for which $k_s=\vert I_s\vert\ge 2$.

In terms of integrals $\tilde f_s$, the Hamiltonian of the system
can be express as
$$
H=\frac12\sum_{s=0}^r \tilde f_s.
$$

\begin{thm} Apart of the integrals arising from the Lax
representation, the rational functions
\begin{equation}\label{Ros}
P_{s,ij}:= P_{ij}, \qquad i,j\in I_s
\end{equation}
 are integrals of the Jacobi-Rosochatius problem
\eqref{RH}. The functions $\tilde f_s, P_s$ are central functions
within the set of integrals $\mathcal F=\{\tilde f_s, P_{s,ij}\}$:
\begin{eqnarray*}
&&  \{\tilde f_{s_1},\tilde f_{s_2}\}_D=0,\qquad \{\tilde
f_{s_1},P_{s_2}\}_D=0, \qquad \{P_{s_1},P_{s_2}\}_D=0,\\
&&  \{\tilde f_{s_1}, P_{s_2,ij}\}_D=0, \qquad \{P_{s_1},
P_{s_2,ij}\}_D=0.
\end{eqnarray*}
Also, the functions $P_{s_1,i_1j_1}$ and $P_{s_2,i_2j_2}$ mutually
commute for distinct $s_1$ and $s_2$,
\begin{equation}\label{cetvrta}
\{P_{s_1,i_1 j_1},P_{s_2,i_2 j_2}\}_D=0.
\end{equation}
\end{thm}

\begin{proof}
The theorem can be verified by a straightforward calculations.
Instead, we consider the Jacobi-Rosochatius problems on an
one-parametric family of deformed, non-symmetric ellipsoids
$$
E_\epsilon^n: \quad \langle A_\epsilon^{-1} x,x\rangle=1, \quad
A_\epsilon=\diag(a_0^\epsilon,\dots,a_n^\epsilon), \quad
a_i^\epsilon\ne a_j^\epsilon, \quad i\ne j,
$$
where
$$
\lim_{\epsilon \rightarrow 0} a_i^\epsilon=a_i,
$$
and $a_i^\epsilon$ are smooth functions defined on some interval
$(-\Delta,\Delta)$.

Let $\{\cdot,\cdot\}_D^\epsilon$ be the associated Dirac-Poisson
bracket \eqref{Dirac_bracket} and let $f_i^\epsilon$ be the
corresponding integrals \eqref{f-JR},
\begin{equation}\label{nula}
\{f_i^\epsilon,f_j^\epsilon\}_D^\epsilon=0.
\end{equation}

Define
\begin{eqnarray*}
&& \tilde f^\epsilon_s=\sum_{i\in I_s} f^\epsilon_i=\sum_{i\in
I_s} \left(y_i^2+\sigma x_i^2+\frac{\mu_i^2}{x_i^2}+\sum_{j\notin
I_s}\frac{P_{ij}}{a^\epsilon_i-a^\epsilon_j}\right), \quad
s=0,\dots,r,\\
&& P_{s,ij}^\epsilon=(a_i^\epsilon-a_j^\epsilon) f_i^\epsilon,
\qquad i,j\in I_s.
\end{eqnarray*}

From \eqref{nula} we get $\{\tilde f^\epsilon_{s_1},\tilde
f^\epsilon_{s_2}\}^\epsilon_D=0$ for all $\epsilon$. Since
$\lim_{\epsilon \rightarrow 0}\tilde f_s^\epsilon=\tilde f_s$, by
taking the limit we obtain
$$
\{\tilde f_{s_1},\tilde f_{s_2}\}_D=0.
$$

On the other hand, $\lim_{\epsilon \rightarrow 0}
P^\epsilon_{s,ij}$ can be singular and depends on the deformation.
Suppose that
$$
\lim_{\epsilon \rightarrow
0}\frac{a_i^\epsilon-a_j^\epsilon}{a_i^\epsilon-a_l^\epsilon}=0,
\qquad l \ne i,j.
$$
Then $\lim_{\epsilon \rightarrow 0} P^\epsilon_{s,ij}=P_{s,ij}$
and from $\{\tilde f^\epsilon_{s_1},
P^\epsilon_{s_2,ij}\}_D^\epsilon=0$ for all $\epsilon$, we get
$$
\{\tilde f_{s_1}, P_{s_2,ij}\}_D=0.
$$

Therefore
$$
\{\tilde f_{s_1}, P_{s_2}\}_D=\sum_{i,j\in I_{s_2}, i<j} \{\tilde
f_{s_1}, P_{s_2,ij}\}_D=0
$$
and
$$
\{ P_{s,ij},H\}_D= \frac12\sum_{s=0}^r \{P_{s,ij}, \tilde
f_s\}_D=0,
$$
that is $P_{s,ij}$ are integrals of the system. Similarly, for
$s_1 \ne s_2$, $i_1,j_1\in I_{s_1}$, $i_2,j_2\in I_{s_2}$, we can
always find a perturbation $A_\epsilon$ such that
$$
\lim_{\epsilon \rightarrow
0}\frac{a_{i_d}^\epsilon-a_{j_d}^\epsilon}{a_{i_d}^\epsilon-a_{l_d}^\epsilon}=0,
\qquad l_d \ne i_d,j_d, \qquad d=1,2.
$$
Therefore, $\lim_{\epsilon \rightarrow 0} P^\epsilon_{s_d,i_d
j_d}=P_{s_d,i_d j_d}$, $d=1,2$, and we have
\begin{eqnarray*}
&& \{P_{s_1,i_1 j_1},P_{s_2,i_2 j_2}\}_D=0,\\
&& \{P_{s_1},P_{s_2,i_2,j_2}\}= \sum_{i_1, j_1\in I_{s_1},
i_1<j_1}
\{P_{s_1,i_1j_1},P_{s_2,i_2 j_2}\}_D=0,\\
 && \{P_{s_1},P_{s_2}\}= \sum_{i_2, j_2\in I_{s_2}, i_2<j_2} \{P_{s_1},P_{s_2,i_2 j_2}\}_D=0.
\end{eqnarray*}

It remains to prove $\{P_{s}, P_{s,ij}\}_D=0$, $\vert I_s\vert \ge
3$. For simplicity, assume $s=0,i=0,j=1$. Consider a deformation
$A_\epsilon$ having the property
\begin{eqnarray}
&&\nonumber  \lim_{\epsilon \rightarrow
0}\frac{a_{0}^\epsilon-a_{1}^\epsilon}{a_{0}^\epsilon-a_{l}^\epsilon}=0,
\qquad l \ne 0,1,\\
&&\label{limes} \lim_{\epsilon \rightarrow
0}\frac{a_{2}^\epsilon-a_{0}^\epsilon}{a_{2}^\epsilon-a_{l}^\epsilon}=0,
\qquad l \ne 0,1,2 \\
&&\nonumber \lim_{\epsilon \rightarrow
0}\frac{a_{2}^\epsilon-a_{0}^\epsilon}{a_{2}^\epsilon-a_{1}^\epsilon}=1.
\end{eqnarray}
For example, we can take $a_0^\epsilon=\alpha_0, \,
a_1^\epsilon=\alpha_0+\epsilon^3, \,
a_2^\epsilon=\alpha_0+\epsilon^3+\epsilon^2$, $
a_3^\epsilon=\alpha_0+\epsilon^3+\epsilon^2+\epsilon,\,\dots,\,
a_{k_0}=\alpha_0+\epsilon^3+\epsilon^2+(k_0-2)\epsilon. $
Subsequently, we get
\begin{eqnarray*}
&& \lim_{\epsilon \rightarrow 0} P^\epsilon_{0,0,1}=\lim_{\epsilon
\rightarrow 0}(a_0^\epsilon-a_1^\epsilon)f^\epsilon_0=P_{0,0,1},
\\
&& \lim_{\epsilon \rightarrow 0} P^\epsilon_{0,2,0}=\lim_{\epsilon
\rightarrow
0}(a_2^\epsilon-a_0^\epsilon)f^\epsilon_2=P_{0,2,0}+P_{0,2,1}=P_{0,0,2}+P_{0,1,2}
\end{eqnarray*}
and, consequently,
\begin{equation}\label{prva}
\{P_{0,0,1},P_{0,1,2}+P_{0,0,2}\}_D=0.
\end{equation}

Next, we take $a_0^\epsilon=\alpha_0, \,
a_1^\epsilon=\alpha_0+\epsilon^2, \,
a_2^\epsilon=\alpha_0+\epsilon^2+\epsilon$,
$a_3^\epsilon=\alpha_0+2\epsilon^2+\epsilon$,
$a_4^\epsilon=\alpha_0+2\epsilon^2+2\epsilon \,\dots,\,
a_{k_0}=\alpha_0+2\epsilon^2+(k_0-2)\epsilon. $ Then
\begin{eqnarray*}
\lim_{\epsilon \rightarrow
0}\frac{a_{0}^\epsilon-a_{1}^\epsilon}{a_{0}^\epsilon-a_{l}^\epsilon}=0,
\qquad l \ne 0,1,\qquad  \lim_{\epsilon \rightarrow
0}\frac{a_{2}^\epsilon-a_{3}^\epsilon}{a_{2}^\epsilon-a_{l}^\epsilon}=0,
\qquad l \ne 2,3,
\end{eqnarray*}
and $ \lim_{\epsilon \rightarrow 0} P^\epsilon_{0,0,1}=P_{0,0,1}$,
$ \lim_{\epsilon \rightarrow 0} P^\epsilon_{0,2,3}=P_{0,2,3}$.
Thus, we get
\begin{equation}\label{druga}
\{P_{0,0,1},P_{0,2,3}\}_D=0.
\end{equation}

Finally, repeating the arguments given for \eqref{prva} and
\eqref{druga}, we obtain the commuting relations
\begin{equation}\label{treca}
\{P_{s,ij},P_{s,ik}+P_{s,jk}\}_D=0, \qquad
\{P_{s,ij},P_{s,kl}\}_D=0, \qquad i,j \ne k,l,
\end{equation}
implying
\begin{eqnarray*}
\{P_{s,ij},P_s\}_D=\sum_{k,l\in I_s,
k<l}\{P_{s,ij},P_{s,kl}\}_D=\sum_{k\in
I_s}\{P_{s,ij},P_{s,ik}\}_D+\{P_{s,ij},P_{s,jk}\}_D=0.
\end{eqnarray*}
\end{proof}

In particular, for the case of the Jacobi problem (where we set
$\mu_0=\dots=\mu_n=0$), or for the case of a geodesic flow on a
symmetric ellipsoid ($\mu_0=\dots=\mu_n=\sigma=0$), the system is
invariant with respect to the $SO(k_0)\times\dots\times SO(k_r)$
action and the integrals \eqref{Ros} reduce to the squares of the
Noether integrals
\begin{equation}\label{Phi}
\Phi_{s,ij}=y_ix_j-x_iy_j, \qquad  i<j, \qquad i,j\in I_s,
\end{equation}
while the central functions $P_s$ reduce to the invariants
\begin{equation}
\Phi^2_s=\sum_{i,j\in I_s,i<j} \Phi^2_{s,ij}.
\end{equation}

The detail study of the geodesic flow  is given in \cite{Kris}
(the case $n=3$ can be found in \cite{DDB}). It is proved that
among the central functions $\tilde f_s, P_s$, with
$\mu_0=\dots=\mu_n=\sigma=0$, there are $r+\rho_0$ independent
ones, while that among $\tilde f_s$ and the Noether integrals
\eqref{Phi} there are $2n-r-\rho$ independent ones. Therefore, for
a sufficiently small parameters $\mu_i$ and $\sigma$, we get that
the dimensions of linear spaces
\begin{eqnarray*}
&& F(x,y)=\langle X_{\tilde f_s}(x,y), X_{P_{s,ij}}(x,y)\,\vert\, s=0,1,\dots,r, \, i,j\in I_s \rangle,\\
&& K(x,y)=\langle X_{\tilde f_s}(x,y), X_{P_{s}}(x,y)\,\vert\,
s=0,1,\dots,r \rangle
\end{eqnarray*}
are at least $2n-r-\rho$ and $r+\rho$, respectively, at a generic
point $(x,y)\in T^*E^n$. Here $X_f$ denotes the Hamiltonian vector
field with respect to the Dirac-Poisson bracket
\eqref{Dirac_bracket}. According to \eqref{moser-rel*} we have the
relation
\begin{equation}\label{peta}
\sum_{s=0}^r \frac{\tilde f_s}{\alpha_s}=\sigma+ \sum_{s=0}^r
\frac{P_s}{\alpha_s^2}+\sum_{i=0}^n\frac{\mu_i^2}{a_i^2},
\end{equation}
so $\dim K(x,y)\le r+\rho$.

Since all object are rational functions, $\dim F(x,y)\ge
2n-r-\rho$ and $\dim K(x,y)= r+\rho$, for a generic values of
$\mu_i$, $\sigma$, $(x,y)\in T^*E^n$. As a result, we conclude
that $\mathcal F=\{\tilde f_s, P_{s,ij}\}$ is a complete set of
integrals (a generic dimension of $F(x,y)$ equals $2n-r-\rho$ and
$\dim F(x,y)+\dim K(x,y)=\dim T^*E^n$) and we can apply the
Nekhoroshev-Mishchenko-Fomenko theorem on noncommutative
integrability (see \cite{N, MF2, AKN}).

\begin{thm}\label{integrabilnost}
The Jacobi-Rosochatius problem \eqref{RH} (the Jacobi problem
\eqref{H}) on a symmetric ellipsoid \eqref{ellipsoid},
\eqref{symmetric} is completely integrable in a non-commutative
sense by means of integrals \eqref{f-JR*} and \eqref{Ros} (where
we set $\mu_i=0$). Generic trajectories take place over
$r+\rho$-dimensional invariant isotropic tori, spanned by the
Hamiltonian vector fields $X_{\tilde f_s}, X_{P_s}$.
\end{thm}

In general, noncommutative integrability implies the usual
Liouville integrability by means of smooth commuting integrals
\cite{BJ2}. Here, the integrals can be chosen to be linear
functions of non-commuting integrals.

\begin{thm} \label{integrabilnost2}
The Jacobi-Rosochatius problem \eqref{RH} on a symmetric ellipsoid
\eqref{ellipsoid}, \eqref{symmetric} is Liouville integrable by
means of integrals \eqref{f-JR*} and
\begin{equation}\label{L-JR}
L_{s,k}=\sum_{i,j \in I_{s,k}, i<j} P_{s,ij}, \qquad
k=1,\dots,k_s-1, \, s=0,\dots,r,
\end{equation}
where
\begin{eqnarray*}
&&  I_{0,k}=\{0,\dots,k\}, \\
&&  I_{1,k}=\{k_0,\dots,k_0+k\},\\
&&  \dots \\
&&I_{r,k}=\{k_0+\dots+k_{r-1},\dots,k_0+\dots+k_{r-1}+k\}.
\end{eqnarray*}
\end{thm}

\begin{proof}
According to \eqref{cetvrta}, we have
$\{L_{s_1,k_1},L_{s_2,k_2}\}_D=0$ for $s_1\ne s_2$, while from
\eqref{treca} we get
$$
\{P_{s,ij},L_{s,k}\}_D=0, \qquad i,j\in I_{s,k},
$$
and, in particular,
$$
\{L_{s,k_1},L_{s,k_2}\}_D=0
$$
(note that $P_s=L_{s,k_{s}-1}$).

It is clear that the integrals \eqref{L-JR} are mutually
independent and their total number is
$$
(k_0-1)+\dots+(k_r-1)=(n+1)-(r+1)=n-r.
$$
On the other hand, from the completeness of $\mathcal F=\{\tilde
f_s, P_{s,ij}\}$, $r$ functions among $\tilde f_s$ are independent
from the integrals $P_{s,ij}$ (we have the relation \eqref{peta}
among them). Therefore, the set of integrals $\{\tilde f_s,
L_{s,k}\}$ is a complete commutative set on $T^*E^n$.
\end{proof}

A choice of commuting integrals is not unique and \eqref{L-JR} is
motivated by a {\it  chain of subalgebras method} in the
construction of commutative functions on Lie algebras (e.g., see
Thimm \cite{Th}). Another complete families of commuting integrals
can be obtained, for example, by using separable variables related
to the degeneration of the elliptic coordinates (see \cite{BKW,
KM}).

\section{Billiards inside ellipsoids}

\subsection{Billiards: continuous and discrete description}
Let $(Q,g)$ be a $n$--dimensional Riemannian manifold and let
$D\subset Q$ be a domain with a (smooth) boundary $\Gamma$. Let
$\pi: T^*Q \to Q$ be a natural projection and let $g^{-1}$ be the
contravariant metric on the cotangent bundle. Consider the {\it
reflection mapping} $ r: \pi^{-1} \Gamma \to \pi^{-1} \Gamma,\quad
\mathbf y_- \mapsto \mathbf y_+$, which associates the covector
$\mathbf y_+\in T^*_\mathbf x Q$, $\mathbf x\in \Gamma$ to a
covector $\mathbf y_-\in T^*_\mathbf x Q$ such that the following
conditions hold:
$$
\vert \mathbf y_+ \vert =\vert \mathbf y_- \vert, \qquad \mathbf
y_+-\mathbf y_- \bot \Gamma.
$$

A {\it billiard} in $D$ is a dynamical system with the phase space
$T^*D$ whose trajectories are geodesics given by the Hamiltonian
equations with the Hamiltonian $H(\mathbf x,\mathbf y)=\frac12
g^{-1}(\mathbf y,\mathbf y)$, reflected at points $\mathbf x\in
\Gamma$ according to the billiard law: $r(\mathbf y_-)=\mathbf
y_+$. Here $\mathbf y_-$ and $\mathbf y_+$ denote the momenta
before and after the reflection. If some potential force field
$V(\mathbf x)$ is added than the system is described with the same
reflection law and Hamiltonian equations with the Hamiltonian
$H(\mathbf x,\mathbf y)=\frac12 g^{-1}(\mathbf y,\mathbf
y)+V(\mathbf x)$.

A function $f: T^*Q \to \R$ is an {\it integral} of the billiard
system if it commutes with the Hamiltonian ($\{f,H\}=0$) and does
not change under the reflection ($f(\mathbf x,\mathbf y)=f(\mathbf
x,r(\mathbf y))$, $\mathbf x\in\Gamma$). The billiard is {\it
completely integrable in the sense of Birkhoff} if it has $n$
independent integrals polynomial in the momenta, which are in
involution (see \cite{Kozlov}).

It is well known that the billiard system within an ellipsoid
$E^{n-1}\subset\R^n$ under the influence of an arbitrary potential
separable in elliptic coordinates is completely integrable in the
sense of Birkhoff \cite{Kozlov, DJ, DJR}. Moreover, the manifold
$T^*D/r\setminus \Sigma$ carries well defined symplectic
structure, such that the billiard flow is the usual Hamiltonian
flow with $n$ commuting smooth integrals ($\Sigma$ is the
codimension two submanifold -- the cotangent bundle of the
ellipsoid, see Lazutkin \cite{Laz}). Whence we can use the
Arnold-Liouville theorem in the description of the system.
Alternatively, we can consider the billiard as a discrete
integrable system with the {\it billiard mapping} $\phi: (\mathbf
x_k,\mathbf y_k)\mapsto (\mathbf x_{k+1},\mathbf y_{k+1})$, where
$\mathbf x_k$ is a sequence of the points of impact and $\mathbf
y_k$ is the corresponding sequence of outgoing momenta \cite{Ves}.
However, it is not a simple task to explicitly describe the
billiard map in the Descartes coordinates, as it the case when the
trajectories between the impacts are the straight lines.

The first result in this direction is performed by Fedorov
\cite{Fe2}, who calculated the billiard map and found the Lax
representation for the billiard system under the influence of the
elastic force \cite{W}. As a slight modification, in this section
we consider the billiard system under the additional influence of
the Rosochatius potential. In the derivation of the billiard
mapping we use a discrete version of the reduction given in
Section 4.

\subsection{Harmonic oscillator and ellipsoidal billiards}
Consider the Jacobi problem \eqref{CP1} on the $2n+1$-dimensional
ellipsoid $E^{2n+1}$. When parameter $a_0$ tends to zero, the
Jacobi flow transforms to the billiard problem within real
$2n-1$-dimensional ellipsoid in $\C^n$
\begin{equation}\label{ellipsoid*}
E^{2n-1}=\{\mathbf z\in\C^{n}\,\vert\, \langle a^{-1} \mathbf
z,\bar{\mathbf z}\rangle=1\},
\end{equation}
where the motion between the impacts is influenced by the elastic
force $-\sigma z$ (harmonic oscillations constrained inside the
ellipsoid \eqref{ellipsoid*}):
\begin{equation}\label{HarO}
\dot{\mathbf z}=\mathbf p, \qquad \dot{\mathbf p}=-\sigma \mathbf
z.
\end{equation}
Here we denoted  $a=\diag(a_1,\dots,a_n)$, $\mathbf
z=(z_1,\dots,z_n)$, $\mathbf p=(p_1,\dots,p_n)$.

If $\sigma \le 0$, then all trajectories have reflections from the
boundary $E^{2n-1}$, while for $\sigma>0$, the initial conditions
$(\mathbf z_0,\mathbf p_0)$ determining the energy $h=H(\mathbf
z_0,\mathbf p_0)$ should satisfy
$$
h+\frac{\sigma}2\langle \mathbf z,\bar{\mathbf
z}\rangle>\epsilon>0, \qquad \mathbf z\in E^{2n-1}.
$$

Consider the integral
\begin{equation}\label{J}
J=\langle A^{-1} p,\bar p\rangle\langle A^{-2} z,\bar z\rangle
-\sigma\langle A^{-2} z,\bar z\rangle
\end{equation}
 of the Jacobi problem \eqref{CP1}.  For distinct $a_i$, it is equal to the sum
$-\sum_i {a_i}^{-2} f_i$. Also, note that $H/J$ for $\sigma=0$
equals to the square of the Joachimsthal integral of the geodesic
flow \cite{Moser}.

In the limiting process, the integral \eqref{J} multiplied by
$a_0$ becomes
\begin{equation}\label{J*}
\lim_{a_0\to 0} a_0\,J=\frac14(\langle a^{-1} \mathbf
z,\bar{\mathbf p} \rangle+\langle a^{-1} \bar{\mathbf z},\mathbf
p\rangle)^2+(1-\langle a^{-1}\mathbf z,\bar{\mathbf
z}\rangle)(\langle a^{-1}\mathbf p,\bar{\mathbf p}\rangle-\sigma).
\end{equation}

Let  $\mathbf z_k$, $k\in\mathbb Z$, be the set of impact points.
By $\mathbf p_k$ we denote the outgoing velocity at $\mathbf z_k$.
From \eqref{J*} we get the integral
$$
J_k=\langle a^{-1} \mathbf z_k,\bar{\mathbf p}_k \rangle+\langle
a^{-1} \bar{\mathbf z}_k,\mathbf p_k\rangle.
$$
of the billiard mapping within ellipsoid \eqref{ellipsoid*}
\begin{equation}\label{BM}
\phi(\mathbf z_k,\mathbf p_k)=(\mathbf z_{k+1},\mathbf p_{k+1}).
\end{equation}

With the above notation, the mapping \eqref{BM} given in
\cite{Fe2} reads
\begin{eqnarray*}
&& \mathbf z_{k+1}=-\frac1{\nu_k}[(\sigma-\langle \mathbf
p_k,a^{-1}\bar{\mathbf p}_k\rangle)\mathbf z_k+J_k \mathbf p_k],\label{Dis1}\\
&&\mathbf p_{k+1}=-\frac1{\nu_k}[(\sigma-\langle \mathbf
p_k,a^{-1}\bar{\mathbf p}_k\rangle)(\mathbf p_k+\pi_k
a^{-1}\mathbf z_k)+ J_k(\pi_k a^{-1}\mathbf p_k-\sigma\mathbf
z_k)],\label{Dis2}
\end{eqnarray*}
where
$$
\nu_k=\sqrt{\sigma J_k^2+(\sigma-\langle\mathbf
p_k,a^{-1}\bar{\mathbf p}_k \rangle)^2}, \quad
\pi_k={J_k}/{\langle \mathbf z_{k+1},a^{-2}\bar{\mathbf
z}_{k+1}\rangle}.
$$

\subsection{The Jacobi-Rosochatius billiard}
The harmonic oscillations with impacts at the boundary
\eqref{ellipsoid*} have integrals
\begin{equation}\label{dis-int}
h_j(\mathbf z_k,\mathbf p_k)=\frac{i}2(z_{k,j}\bar
p_{k,j}-p_{k,j}\bar z_{k,j}), \quad j=1,\dots,n,
\end{equation}
and we can perform a discrete analogue of the reduction described
in Section 4.

Let us fix the values of the integrals
$$
h_j=\mu_j,\quad j=1,\dots,n,
$$
where some of $\mu_j$ can be equal to zero.  Introduce the
coordinate change
\begin{equation}\label{polar-dis2}
z_{k,j}=x_{k,j}e^{i\varphi_{k,j}}, \quad p_{k,j}=y_{k,j}
e^{i\varphi_{k,j}}+i\frac{\mu_j}{x_{k,j}} e^{i\varphi_{k,j}} \quad
\text{for}\quad \mu_j\ne 0
\end{equation}
and consider the restrictions of $z_{k,j},p_{k,j}$ to $\R$,
$$
z_{k,j}=x_{k,j}\in\R, \quad p_{k,j}=y_{k,j}\in\R \quad
\text{for}\quad \mu_j=0.
$$

The mapping \eqref{BM}, induces the mapping
\begin{equation}\label{BM*}
\Phi(\mathbf x_k,\mathbf y_k)=(\mathbf x_{k+1},\mathbf y_{k+1}),
\qquad \mathbf x_k, \mathbf x_{k+1}\in E^{n-1}_*
\end{equation}
where $E^{n-1}_*$ is the ellipsoidal component of the boundary of
the domain
$$
D^n=\{\mathbf x\in\R^n\,\vert\, \langle \mathbf x,a^{-1}\mathbf
x\rangle \le 1, \,\, x_j \ge 0\,\,\text{for}\,\, \mu_j\ne 0,\,
j=1,\dots,n\}.
$$

\begin{figure}[ht]
\includegraphics{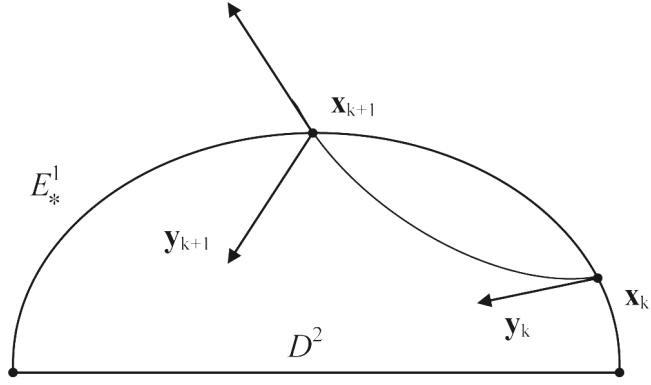}
\caption{Billiard domain for $n=2$, $\mu_1=0$, $\mu_2 \ne 0$.}
\end{figure}

\begin{lem}
\begin{eqnarray*}
&& x_{k+1,j}=\frac{1}{\nu_k}\sqrt{\left(J_k
y_{k,j}+K_k x_{k,j}\right)^2+\frac{\mu_j^2}{x_{k,j}^2}}, \qquad \mu_j\ne 0, \\
&& x_{k+1,j}=-\frac{1}{\nu_k}\left[K_k x_{k,j}+J_k y_{k,j}\right],
\qquad \qquad\qquad \mu_j=
0,\\
&&
y_{k+1,j}=-\frac{e^{-i\delta_{k,j}}}{\nu_k}\left[K_k\left(\frac{\pi_k}{a_j}x_{k,j}+y_{k,j}+i\frac{\mu_j}{x_{k,j}}
\right) \right]\\
&&\qquad\qquad-\frac{J_k
e^{-i\delta_{k,j}}}{\nu_k}\left[\frac{\pi_k}{a_j}y_{k,j}-\sigma
x_{k,j}+i\frac{\pi_k\mu_j}{a_jx_{k,j}}\right]-i\frac{\mu_j}{x_{k+1,j}},
\end{eqnarray*}
where
\begin{eqnarray*}
&& J_k=2{\langle \mathbf x_{k},a^{-1}{\mathbf y}_{k}\rangle},\\
&& K_k=\sigma-\langle \mathbf y_k,a^{-1}\mathbf
y_k\rangle-\left\langle \frac{\mathbf\mu}{\mathbf
x_k},a^{-1}\frac{\mathbf\mu}{\mathbf x_k}\right\rangle,\\
 && \nu_k=\sqrt{\sigma J_k^2 + K_k^2},\\
&& \pi_k=J_k/{\langle \mathbf x_{k+1},a^{-2}{\mathbf
x}_{k+1}\rangle},\\
&& \delta_{k,j}=Arg\left[-K_k x_{k,j}-J_k
y_{k,j}-iJ_k\frac{\mu_j}{x_{k,j}}\right].
\end{eqnarray*}
\end{lem}

This is a billiard mapping of the billiard system within a domain
$D^n$, where the motion between the impacts is described by the
Hamiltonian function
$$
H(\mathbf x,\mathbf y)=\frac12\langle\mathbf y,\mathbf
y\rangle+\frac{\sigma}{2}\langle\mathbf x,\mathbf
x\rangle+\frac12\sum_{j=1}^n\frac{\mu_k^2}{x_j^2}
$$
that is,
\begin{equation}\label{JR}
\dot{\mathbf x}=\mathbf y, \qquad \dot{\mathbf x}=-\sigma \mathbf
x+\frac{\mu^2}{\mathbf x^3},
\end{equation}
where $\mu^2/\mathbf x^3=(\mu_1^2/x_1^3,\dots,\mu_n^2/x_n^3)$.

The trajectories $\mathbf x(t)$ of \eqref{JR} are projections of
trajectories $\mathbf z(t)$ of the harmonic oscillator
\eqref{HarO} (conics or degenerate conics) by the reduction
\eqref{polar-dis2}. As above, if  $\sigma \le 0$, then all
trajectories will have reflections from the boundary, while for
$\sigma>0$, the initial conditions $(\mathbf x_0,\mathbf p_0)$
determining the energy $h=H(\mathbf x_0,\mathbf y_0)$ should
satisfy
$$
h+\frac{\sigma}2\langle \mathbf x,{\mathbf x}\rangle>\epsilon>0,
\qquad \mathbf x\in E^{n-1}_*.
$$

We refer to \eqref{BM*} as the {\it Jacobi-Rosochatius billiard
mapping}. After a straightforward modification of Fedorov's  Lax
representation \cite{Fe2} by applying a discrete reduction, we
obtain the following statement.

\begin{thm} \label{C3} Suppose that the eigenvalues $a_i$ of the matrix $a$
are distinct. Up to the action of the group  generated by the
reflections
\begin{equation}\label{simetrije}
(x_j,y_j)\mapsto (s_jx_j,s_jy_j), \quad s_j=\pm 1, \quad
{\mu_j=0}, \quad s_j=1, \quad \mu_j\ne 0,
\end{equation}
$j=1,\dots,n,$ the Jacobi-Rosochatius billiard map \eqref{BM*} is
equivalent to the matrix equation
\begin{equation}\label{LAF}
{\mathcal L}_{\mathbf x_{k+1},\mathbf y_{k+1}}(\lambda)=\mathcal
A_{\mathbf x_{k},\mathbf y_{k}}(\lambda)\mathcal L_{\mathbf
x_{k},\mathbf y_{k}}(\lambda)\mathcal A^{-1}_{\mathbf
x_{k},\mathbf y_{k}}(\lambda)
\end{equation}
with $2\times 2$ matrices depending on the parameter $\lambda$
\begin{eqnarray}
&& \label{LAF*} \mathcal L_{\mathbf x_k,\mathbf y_k}(\lambda)=
\begin{pmatrix}
Q_\lambda(\mathbf x_k,\mathbf y_k) & Q_\lambda(\mathbf y_k,\mathbf y_k)+Q_\lambda(\frac{\mu}{\mathbf x_k},\frac{\mu}{\mathbf x_k})+\sigma \\
-1-Q_\lambda(\mathbf x_k,\mathbf x_k) & -Q_\lambda(\mathbf
y_k,\mathbf x_k)
\end{pmatrix},\\
&& \nonumber \mathcal A_{\mathbf x_k,\mathbf y_k}(\lambda)=
\begin{pmatrix}
K_k\lambda+J_k\pi_k & \sigma J_k\lambda-K_k \pi_k
\\
-J_k\lambda & K_k\lambda
\end{pmatrix},
\end{eqnarray}
where $Q_\lambda(\cdot,\cdot)=\langle a^{-1}_\lambda
\cdot,\cdot\rangle$,
$a_\lambda=\diag(\lambda-a_1,\dots,a_n-\lambda)$.
\end{thm}

\begin{rem}
Note that the billiard can be seen as a limit of the system
\eqref{RH} as $a_0$ tends to zero and $\mu_0=0$. Contrary, as the
value of the integral $J_k=2\langle \mathbf x_k,\mathbf
y_k\rangle$ tends to zero, the impact points $\mathbf x_k$
approximate the trajectories of the Jacobi-Rosochatius problem on
$E^{n-1}_*$. As in the case of the Jacobi-Rosochatius system
\eqref{RH} on a symmetric ellipsoid, the Jacobi-Rosochatius
billiard map \eqref{BM*} for a symmetric ellipsoid (see the
equation \eqref{symmetric2} given below) is an example of a
discrete system integrable in a noncommutative sense.
\end{rem}

\subsection{The Chasles and Poncelet theorems} In what follows we
shall give a geometric interpretation of the integrability.

Let us recall on a well known variant of the {Chasles} theorem for
the billiard system within ellipsoid
\begin{equation}\label{E3}
E^{n-1}=\{\mathbf x\in\R^n\,\vert\, \langle \mathbf
x,a^{-1}\mathbf x\rangle =1\}
\end{equation}
without external forces, e.g., see \cite{Kozlov, Ves, DrRa}.
Assume that the eigenvalues of the matrix $a$ are different.  Let
$\mathbf x_k$ be a generic sequence of impact points. Then the
sequence of lines $\mathbf x_k\mathbf x_{k+1}$ is simultaneously
tangent to the same set of quadrics $\mathcal
Q_{\eta_1},\dots,\mathcal Q_{\eta_{n-1}}$ from the confocal family
\begin{equation}\label{confocal2}
\mathcal Q_\lambda: \qquad
\sum_{i=1}^n\frac{x_i^2}{a_i-\lambda}=1.
\end{equation}

Namely, the set of lines $p=p(\mathbf x,\mathbf y)=\{\mathbf
x+s\mathbf y\,\vert\, s\in\R\}$ that are tangent to the quadric
$\mathcal Q_\eta$ from the confocal family \eqref{confocal2} are
given by the equation \cite{Moser}
\begin{equation}\label{MC}
\Phi_{\mathbf x,\mathbf y}(\eta)=(Q_\eta(\mathbf x,\mathbf
x)+1)Q_\eta(\mathbf y,\mathbf y)-Q_\eta^2(\mathbf x,\mathbf y)=0.
\end{equation}

On the other hand, from \eqref{LAF} where we set $\sigma=0$ and
$\mu_i=0$,  $\Phi_{\mathbf x_k,\mathbf y_k}(\lambda)=\det\mathcal
L_{\mathbf x_{k},\mathbf y_{k}}(\lambda)$ is an integral of the
system. Therefore, if $\eta$ is a zero of $\Phi_{\mathbf
x_{k},\mathbf y_{k}}(\lambda)$, then the lines $p_k=p(\mathbf
x_k,\mathbf y_k)=\mathbf x_k\mathbf x_{k+1}$, $k\in\mathbb Z$ are
simultaneously tangent to $\mathcal Q_{\eta}$. For a generic
trajectory $(\mathbf x_k,\mathbf y_k)$, we have $n-1$ distinct
solutions $\eta_1,\dots,\eta_{n-1}$ of the equation $\Phi_{\mathbf
x_k, \mathbf y_k}(\lambda)=0$.

Moreover, suppose that the trajectory $\mathbf x_k$ is periodic.
Then any billiard trajectory which shares the same caustic
quadrics is also periodic, with the same period ({\it the Poncelet
theorem} in $\R^n$ \cite{P, CCS, DrRa2, DrRa}).

An analytical condition on caustics $\mathcal
Q_{\eta_1},\dots,\mathcal Q_{\eta_{n-1}}$ for periodic billiard
trajectories is derived by Dragovi\'c and Radnovi\'c, generalizing
classical Cayley's condition for $n=2$  \cite{DrRa2, DrRa}. The
geometry of the lines common to the confocal quadrics is further
studied in \cite{Knorr, DrRa}, while Chasles's-type theorems for
several natural mechanical systems are given in \cite{Mum, Fe,
FeJo}.

Let $d$ be the number of indexes $i$ for which $\mu_i\ne 0$.

\begin{thm}\label{C4}
Assume that the eigenvalues of the matrix $a$ are different. Let
$(\mathbf x_k,\mathbf y_k)$ be a generic trajectory of the
Jacobi-Rosochatius billiard map \eqref{BM*}.

(i) Assume $\sigma\ne 0$ (respectively, $\sigma=0$). The
trajectories
\begin{equation}\label{C5}
l_k: \qquad \mathbf x_k(t), \qquad t\in\R
\end{equation}
of the Jacobi-Rosochatius system \eqref{JR} with initial
conditions $(\mathbf x_k(0),\mathbf y_k(0))=(\mathbf x_k,\mathbf
y_k)$ are simultaneously tangent to the quadrics
\begin{equation}\label{C6}
\mathcal Q_{\eta_1},\dots,\mathcal Q_{\eta_{n+d}} \qquad
(\text{respectively}, \, \mathcal Q_{\eta_1},\dots,\mathcal
Q_{\eta_{n-1+d}})
\end{equation}
from the confocal family \eqref{confocal2}, where
$\eta_1,\dots,\eta_{n+d}$ (respectively,
$\eta_1,\dots,\eta_{n-1+d}$) are solutions of the equation
$\det\mathcal L_{\mathbf x_k, \mathbf y_k}(\lambda)=0$.

(ii) Suppose that the trajectory $(\mathbf x_k, \mathbf y_k)$ is
periodic. Then any billiard trajectory which shares the same
caustic quadrics \eqref{C6} is also periodic with the same period.
\end{thm}

\begin{proof} (i) Firstly, we take $\sigma=0$. For a sake of simplicity, assume $\mu_1,\dots,\mu_d
\ne 0, \, \mu_{d+1}=\dots=\mu_n=0$.

Consider the billiard system without external forces within a
symmetric ellipsoid $E^{n+d-1}$ in
$$
\R^{n+d}\cong\C^d(z_1,\dots,z_d)\times\R^{n-d}(x_{d+1},\dots,x_d),
$$
where we take
$a=(a_1,a_1,\dots,a_d,a_d,a_{d+1},a_{d+2},\dots,a_n)$, $a_i\ne
a_j$, $i\ne j$. Let
\begin{equation}\label{lift}
(\tilde{\mathbf x}_k,\tilde{\mathbf
y}_k)=(z_{1,k},\dots,z_{d,k},x_{d+1,k},\dots,x_{n,k},
p_{1,k},\dots,p_{d,k},y_{d+1},\dots,y_{n,k})
\end{equation}
be a generic billiard trajectory with the values of the integrals
\begin{equation}\label{lift*}
h_j=\frac{i}2(z_{k,j}\bar p_{k,j}-p_{k,j}\bar z_{k,j})=\mu_j,
\quad j=1,\dots,d.
\end{equation}

From Lemma \eqref{CTS} below, the lines $\tilde p_k=\tilde{\mathbf
x}_k\tilde{\mathbf x}_{k+1} \subset \C^d\times\R^{n-d}$ determined
by the billiard trajectory \eqref{lift} are tangent to the
$N=n+d-1$ quadrics
\begin{equation}\label{KR}
\tilde{\mathcal Q}_{\eta_i}: \,
\frac{z_1^2}{a_1-\eta_l}+\dots+\frac{z_d^2}{a_d-\eta_l}+
\frac{x_{d+1}^2}{a_{d+1}-\eta_l}+\dots+\frac{x_n^2}{a_n-\eta_l}=1,
\end{equation}
where $\eta_l$ are zeros of the corresponding polynomial
\eqref{MC*}.

By applying the reduction \eqref{polar-dis2} in the variables
$z_1,\dots,z_d$ we get the billiard trajectory $(\mathbf
x_k,\mathbf y_k)$ of the Rosochatius billiard map \eqref{BM*} with
$\sigma=0$, $\mu_1,\dots,\mu_d \ne 0, \, \mu_{d+1}=\dots=\mu_n=0$.
At the same time,
 $\det\mathcal L_{\mathbf x_k,\mathbf y_k}(\eta_i)=0$,
 $i=1,\dots,n+d-1$ and the lines $\tilde p_k$ project to the
 curves \eqref{C5} tangent to the quadrics \eqref{C6}.

In the other direction, let $(\mathbf x_k,\mathbf y_k)$ be a
trajectory of the Rosochatius billiard map \eqref{BM*} with
$\sigma=0$, $\mu_1,\dots,\mu_d \ne 0, \, \mu_{d+1}=\dots=\mu_n=0$.
Then we can lift $(\mathbf x_k,\mathbf y_k)$ to the
$SO(2)^d$-invariant family of billiard trajectories \eqref{lift}
of the billiard system without external forces within $E^{n+d-1}$
satisfying \eqref{lift*}.

The case $\sigma\ne 0, \mu_1=\dots=\mu_n=0$ is proved by Fedorov
\cite{Fe2}. For $\sigma\ne 0$, instead of the equation \eqref{MC},
which characterizes tangent lines to the quadric $\mathcal
Q_\eta$, we use the equation
\begin{equation}\label{MCF} \Phi_{\mathbf x,\mathbf y}^\sigma(\eta)=(Q_\eta(\mathbf
x,\mathbf x)+1)(Q_\eta(\mathbf y,\mathbf
y)+\sigma)-Q_\eta^2(\mathbf x,\mathbf y)=0.
\end{equation}

A conic $l=\{\mathbf x(t)\,\vert\,t\in\R\}$ associated to a
solution of the equation
$$
\dot{\mathbf x}(t)=\mathbf y, \qquad \dot{\mathbf
y}(t)=-\sigma\mathbf x
$$
with the initial condition $(\mathbf x_0,\mathbf y_0)$ is tangent
to $\mathcal Q_\eta$ if and only if $\Phi_{\mathbf x,\mathbf
y}^\sigma(\eta)=0$ (Proposition 1 in \cite{Fe2}). The statement
can be applied for matrices $a$ with multiple eigenvalues and
repeating the arguments given above, item (i) follows.

\

(ii) Suppose that the trajectory $(\mathbf x_k, \mathbf y_k)$ with
the energy $H=h$ is periodic. The Jacobi-Rosochatius billiard
\eqref{BM*} can be considered as a usual integrable Hamiltonian
system on $T^*D/r\setminus \Sigma$ (we take only trajectories
which are not tangent to the boundary $E^{n-1}_*$, that is
$J_k=2\langle \mathbf x_k,\mathbf y_k\rangle\ne 0$), see
\cite{Laz}. The associated continuous billiard trajectory
$(\mathbf x(t),\mathbf y(t))$ is periodic as well.

The parameters of caustics $\eta_1,\dots,\eta_{n+d}$
(respectively, $\eta_1,\dots,\eta_{n+d-1}$) correspond to the
unique values $c_1,\dots,c_n$, $c_1+\dots+c_n=2h$ of the integrals
\eqref{f-JR} (where we set $i,j=1,\dots,n$) determined from the
equations
$$
\det\mathcal L_{\mathbf x(t),\mathbf
y(t)}(\eta_l)=\sigma+\sum_{i=1}^n
\frac{c_i}{\eta_l-a_i}+\sum_{i=1}^n\frac{\mu_i^2}{(\eta_l-a_i)^2}=0.
$$

The curves \eqref{C5} tangent to the caustics \eqref{C6} are
associated to billiard trajectories which
 belong to the
Lagrangian tori that are components of the invariant level set
$$
M_{c_1,\dots,c_n} \subset T^*D/r\setminus \Sigma: \qquad
f_1=c_1,\dots,f_n=c_n.
$$

According to the Arnold-Liuville theorem, if a trajectory on a
regular Lagrangian torus is periodic, all trajectories on the same
torus are also periodic with the same period. Since all components
of $M_{c_1,\dots,c_n}$ are related by the reflections
\eqref{simetrije}, it is clear that if one Lagrangian torus is
periodic, so are the others one.
\end{proof}

We note that there is a natural generalization of the presented
results to the ellipsoidal billiards on spheres and hyperbolic
spaces \cite{Kozlov, Ves2, DJR, Ta}.

\begin{lem}\label{CTS}{\rm [The Chasles theorem for a symmetric
ellipsoid]} Consider a billiard within a symmetric ellipsoid
$E^{n-1}$ defined by
\begin{eqnarray}
&& \label{symmetric2} a_i=\alpha_1, \qquad i\in I_1=\{1,\dots,k_1\},\,\dots\,, \\
&& \nonumber a_i=\alpha_r, \qquad i\in
I_r=\{k_1+\dots+k_{r-1}+1,\dots,k_0+\dots+k_{r-1}+k_r\},
\end{eqnarray}
$\alpha_i\ne\alpha_j, i\ne j$, $k_1+\dots+k_r=n$. Let $\mathbf
x_k$ be a generic sequence of impact points. Then the sequence of
lines $\mathbf x_k\mathbf x_{k+1}$ is simultaneously tangent to
the same set of quadrics $\mathcal Q_{\eta_1},\dots,\mathcal
Q_{\eta_{N}}$ from the confocal family \eqref{confocal2}, where
$N=\delta_1+\dots+\delta_r-1$, $\delta_s=2$ for $k_s=\vert I_s
\vert \ge 2$, and $\delta_s=1$ for $k_s=\vert
I_s\vert=1$.\footnote{Note that $N=r+\rho-1$, where $r+\rho$ is
the dimension of invariant isotropic tori of the corresponding
geodesic flow on the symmetric ellipsoid \eqref{ellipsoid} in
$\R^{n+1}$ with semi-axis $\sqrt{a_0}\approx 0$ and
$\sqrt{a_1},\dots,\sqrt{a_n}$ given by \eqref{symmetric2} (Theorem
\ref{integrabilnost}).} Further, due to the
$SO(k_1)\times\dots\times SO(k_r)$ symmetry of the system, if
$\mathbf x_k$ is a billiard trajectory, so is $R(\mathbf x_k)$,
$R\in SO(k_1)\times\dots\times SO(k_r)$, and the lines $R(\mathbf
x_k\mathbf  x_{k+1})$ are tangent to the same set of $N$ quadrics.
\end{lem}

\begin{proof}
Note that the description \eqref{MC} of the tangent lines to the
quadric $\mathcal Q_\eta$ does hold for the matrices $a$ with
multiple eigenvalues \eqref{symmetric2}. According to
\eqref{moser-rel*}, the number of quadrics tangent to a generic
line $p=p(\mathbf x,\mathbf y)$ equals to the number of zeros of
the polynomial
\begin{eqnarray}
\nonumber \Psi_{\mathbf x,\mathbf y}(\lambda) &=&
 (\lambda-\alpha_1)^{\delta_1} \cdots
(\lambda-\alpha_r)^{\delta_r}\det\mathcal L_{\mathbf x,\mathbf
y}(\lambda)\\
\label{MC*}  &=& \sum_{s=1}^r
\left((\lambda-\alpha_s)^{\delta_s-1}\prod_{i\ne s}
(\lambda-\alpha_i)^{\delta_i}\tilde f_s+\prod_{i\ne s}
(\lambda-\alpha_i)^{\delta_i}P_s\right),
\end{eqnarray}
where $\tilde f_s$, $P_s$ are given by \eqref{f-JR*}, \eqref{P-s}
(where we set $\sigma=0$, $\mu_i=0$). Therefore, a generic line
$p=p(\mathbf x,\mathbf y)$ is tangent to
$N=\delta_1+\dots+\delta_r-1$ quadrics from the family
\eqref{confocal2}.

Consider a generic billiard trajectory $(\mathbf x_k,\mathbf y_k)$
within a symmetric ellipsoid \eqref{E3}, \eqref{symmetric2}. As
above, if the line $p_k=p(\mathbf x_k,\mathbf y_k)=\mathbf
x_k\mathbf x_{k+1}$ is tangent to the quadric $\mathcal Q_{\eta}$
defined by the equation \eqref{confocal2}, i.e, $\Psi_{\mathbf
x_{k},\mathbf y_{k}}(\eta)=0$,  then the lines $p_k$, $k\in\mathbb
Z$ are also tangent to $\mathcal Q_\eta$.
\end{proof}

\section{Hierarchy of the Lax representations}

By taking the limit $a_0 \rightarrow 0$ and assuming $\mu_0=0$,
from Theorem \ref{T6} we can write down the Lax representation
\begin{eqnarray*} && \dot{\mathcal L}_{\mathbf x,\mathbf
y}(\lambda)=[\mathcal L_{\mathbf x,\mathbf y}(\lambda),\mathcal
A_{\mathbf x,\mathbf y}(\lambda)],
 \\
 && \mathcal L_{\mathbf x,\mathbf y}(\lambda)=
\begin{pmatrix}
Q_\lambda(\mathbf x,\mathbf y) & Q_\lambda(\mathbf y,\mathbf y)+Q_\lambda(\frac{\mu}{\mathbf x},\frac{\mu}{\mathbf x})+\sigma \\
-1-Q_\lambda(\mathbf x,\mathbf x) & -Q_\lambda(\mathbf y,\mathbf
x)
\end{pmatrix}, \\
&& \mathcal A_{\mathbf x,\mathbf y}(\lambda)=
\begin{pmatrix}
0 & -\sigma \\
1 & 0
\end{pmatrix}.
\end{eqnarray*}
for the Jacobi-Rosochatius system on $\R^n$ \eqref{JR}, i.e, for
the trajectories of the billiard system between the impacts. Here
we used
\begin{eqnarray*}
\frac{\sigma-\langle A^{-1}y,y\rangle-\langle
A^{-1}\frac{\mu}x,\frac{\mu}x\rangle}{\langle x,A^{-2}
x\rangle}=\frac{\sigma-\langle a^{-1}\mathbf y,\mathbf
y\rangle-\frac{a_0(1-2\langle \mathbf x,a^{-1}\mathbf
y\rangle)^2}{4(1-\langle \mathbf x,a^{-1}\mathbf
x\rangle)}-\langle a^{-1}\frac{\mu}{\mathbf x},\frac{\mu}{\mathbf
x}\rangle}{\langle \mathbf x,a^{-2} \mathbf
x\rangle+\frac{1}{a_0}(1-\langle \mathbf x,a^{-1}\mathbf
x\rangle)} \rightarrow 0,
\end{eqnarray*}
as $ a_0 \rightarrow 0$, where $x=(x_0,\mathbf x)$ and
$y=(y_0,\mathbf y)$ satisfy the constraints \eqref{E-constraints}.

This Lax representation, for $\mu_1=\dots=\mu_n$, is equivalent to
the Lax representation for the harmonic oscillator given in
\cite{EEKT}, where, by induction, the Lax representations for
natural mechanical systems with polynomial potentials separable in
elliptic coordinates in $\R^n$ are given  (see also \cite{AW}). In
the same way, by using Theorem \ref{T6}, we can give the $2\times
2$-Lax representations for the separable potential perturbations
of the Jacobi-Rosochatius system on $E^n$.

\subsection{Separable potentials in $\R^{n+1}$}
Recall, a potential $V(x)$ is separable in the elliptic
coordinates $\lambda_0 < a_0 < \lambda_1 < a_1 < \dots < \lambda_n
< a_n$ defined by \eqref{confocal} if and only if it is a solution
of the Bertrand-Darboux equations
\begin{equation}\label{BD}
(a_i-a_j)\frac{\partial^2 V}{\partial x_i \partial x_j}+ \left(
x_i\frac{\partial}{\partial x_j}-  x_j\frac{\partial}{\partial
x_i}\right) \left(2V+\sum_{k=0}^{n+1} x_k \frac{\partial
V}{\partial x_k}\right)=0, \quad i\ne j
\end{equation}
(Benenti \cite{Be}, see also Marshall and Wojciechowski
\cite{MW}). The solutions of \eqref{BD} can be written in the form
$V(x)=\sum_i x_i \frac{\partial \mathcal V}{\partial x_i}$, where
$\mathcal V(x)$ are solutions of
\begin{equation*}\label{ZZ}
(a_i-a_j)\frac{\partial^2 \mathcal V}{\partial x_i \partial
x_j}=\left( x_i\frac{\partial}{\partial x_j}-
x_j\frac{\partial}{\partial x_i}\right) \left(\sum_k
x_k\frac{\partial \mathcal V}{\partial x_k} \right), \quad i\ne j
\end{equation*}
(see \cite{Za}). Then, a complete set of commuting integrals is
given by
\begin{equation}\label{novi int}
f_i=y_i^2+\sum_{j\ne i}\frac{(y_ix_j-y_jx_i)^2}{a_i-a_j}+2F_i(x),
\qquad i=0,1,\dots,n,
\end{equation}
where $F_i(x)=x_i\frac{\partial \mathcal V}{\partial x_i}$
\cite{MW, Za}.

Polynomial potentials are described in \cite{V, Za, KBM}. Basic
homogeneous polynomial solutions $V^{(k)}$ of degree $2k$ of the
equations \eqref{BD} in elliptic coordinates reads
$$
V^{(k)}(\lambda_0,\dots,\lambda_n)=-\sum_{j=0}^{n+1}\frac{\lambda_j^{k-1}\prod_i(\lambda_j-a_i)}{\prod_{i\ne
j}(\lambda_j-\lambda_i)}
$$
and $V^{(k)},F_0^{(k)},\dots,F_n^{(k)}, \, k\in\mathbb N$ satisfy
the system of the recurrence relations
\begin{equation}\label{rec}
F_i^{(k+1)}=a_i F_i^{(k)}-x_i^2 V^{(k)}, \quad F_i^{(1)}=x_i^2,
\quad V^{(k)}(x)=F^{(k)}_0+\dots+F^{(k)}_n
\end{equation}
(we use the notation given by Zaitsev \cite{Za}). For example,
\begin{eqnarray*}
&& V^{(1)}=\langle x,x\rangle \quad (\text{the Hook potential}), \quad F^{(1)}_i=x_i^2,  \\
&& V^{(2)}=\langle Ax,x\rangle- V^{(1)}\langle x,x\rangle \quad
(\text{the Garnier potential}), \quad
F^{(2)}_i=x_i^2(a_i-V^{(1)}),\\
&& V^{(3)}=\langle A^2 x,x\rangle- V^{(1)}\langle A x,x\rangle -
V^{(2)} \langle x,x\rangle, \quad F^{(3)}_i=x_i^2(a^2_i-a_i
V^{(1)}-V^{(2)})).
\end{eqnarray*}

The rational and the Laurent polynomial solutions of \eqref{BD}
are given in \cite{V, KBM, Dr, DJ}. The basis for degrees $-2$ and
$-4$ are given by
\begin{eqnarray*}
&& V^{(-1)}_s=\frac{1}{x_s^2}\quad (\text{the Rosochatius
potentials}),\\
&& V^{(-2)}_s(x)=\frac{1}{x_s^4} \left(1+\sum_{j\ne
s}\frac{x_j^2}{a_s-a_j}\right),
 \quad s=0,\dots,n, \\
&& F^{(-1)}_{s,i}=\frac{1}{a_i-a_s}\frac{x_i^2}{x_s^2}, \quad
F^{(-2)}_{s,i}=\frac{2}{a_i-a_s}\frac{x_i^2}{x_s^4}\left(1+\sum_{j\ne
s}\frac{x_j^2}{a_s-a_j}\right), \quad i\ne s, \\
&& F^{(-k)}_{s,s}=V^{(-k)}_s-\sum_{i\ne s} F^{(-k)}_{s,i}, \qquad
k=1,2.
\end{eqnarray*}

\subsection{Natural mechanical systems on ellipsoids with separable
potentials} Consider the motion of a material point on an
ellipsoid \eqref{ellipsoid} under the influence of the potential
\begin{equation}\label{sep-pot}
V(x)=V^+(x)+\frac12\sum_{i=0}^n\frac{\mu_i^2}{x_i^2}, \quad
V^+(x)=\frac12\sum_{k=1}^m \sigma_k V^{(k)}(x),
\end{equation}
where $\sigma_k$ are real parameters. The equations of motion are
\begin{equation}\label{H-sep-pot}
\dot x=y, \quad \dot y=-\frac{\langle A^{-1}y,y\rangle-\langle
\nabla V(x),A^{-1}x\rangle}{\langle A^{-2}x,x\rangle}
A^{-1}x-\nabla V(x).
\end{equation}

As a straightforward generalization of Theorem \ref{T6}, by using
the constraints \eqref{E-constraints}, the recurrence relations
\eqref{rec}, and the identities
$$
A_\lambda A^k=(\lambda-A)^{-1}A^k=\lambda^k
A^{-1}_\lambda-\sum_{i=0}^{k-1} \lambda^{k-i} A^i, \qquad
k\in\mathbb N
$$
we get.

\begin{thm} Suppose that the eigenvalues $a_i$ of the matrix $A$
are distinct. Up to the action of the group $\mathbb Z_2^{n+1}$
generated by the reflections
$$
(x_i,y_i)\longmapsto (s_ix_i,s_iy_i), \quad s_i=\pm 1, \qquad
i=0,1,\dots,n,
$$
the system \eqref{H-sep-pot} is equivalent to the matrix equation
\begin{equation}
\dot{\mathcal L}(\lambda)=[\mathcal L(\lambda),\mathcal
A(\lambda)]
\end{equation}
with $2\times 2$ matrices depending on the parameter $\lambda$
\begin{eqnarray*}
&& \mathcal L(\lambda)=
\begin{pmatrix}
q_\lambda(x,y) & q_\lambda(y,y)+q_\lambda(\frac{\mu}{x},\frac{\mu}{x})+\Delta(x,\lambda) \\
-1-q_\lambda(x,x) & -q_\lambda(y,x)
\end{pmatrix},\\
&& \mathcal A(\lambda)=\frac{1}{\langle A^{-2}x,x\rangle}
\begin{pmatrix}
0 & \frac{1}{\lambda}(\langle \nabla V^+(x),A^{-1}x
\rangle-\langle A^{-1}y,y\rangle-\langle
A^{-1}\frac{\mu}x,\frac{\mu}x\rangle)-\Omega(x,\lambda)\\ \langle
A^{-2}x,x\rangle & 0
\end{pmatrix},
\end{eqnarray*}
where
\begin{eqnarray*}
&&\Delta(x,\lambda)=\sigma_1\Delta_1(x,\lambda)+\dots+\sigma_m\Delta_m(x,\lambda),\\
&&\Delta_k(x,\lambda)=\lambda^{k-1}-\lambda^{k-2}V^{(1)}-\lambda^{k-3} V^{(2)}-\dots-\lambda V^{(k-2)}-V^{(k-1)},\\
&&\Omega(x,\lambda)=\langle A^{-2}
x,x\rangle\left(\sigma_1\Omega_1(x,\lambda)+\dots+\sigma_m
\Omega_m(x,\lambda)\right),
\end{eqnarray*}
and $\Omega_k(x,\lambda)$ are determined from the equations
$$
2\Omega_k(1+q_\lambda(x,x))=2\Delta_k(x,\lambda)+\langle
A^{-1}_\lambda x,\nabla V^{(k)}(x)\rangle, \qquad k=1,\dots,m.
$$
\end{thm}

For example,
\begin{eqnarray*}
&&\Delta_1=1, \qquad \qquad \qquad \Omega_1=1,\\
&&\Delta_2=\lambda-\langle x,x\rangle, \qquad\,\,
\Omega_2=\lambda-2\langle x,x\rangle,\\
&&\Delta_3=\lambda^2-\lambda\langle x,x\rangle-\langle
Ax,x\rangle+\langle x,x\rangle^2, \\
&&\Omega_3=\lambda^2-2\lambda\langle x,x\rangle-2\langle
Ax,x\rangle +3\langle x,x\rangle^2.
\end{eqnarray*}

Now, the integrals \eqref{novi int} and the Lax representation are
related by
\begin{eqnarray*}
\det(\mathcal L(\lambda)) &=&
(1+q_{\lambda}(x,x))\left(q_\lambda(y,y)+q_\lambda\left(\frac{\mu}{x},
\frac{\mu}{x}\right)+\Delta(x,\lambda)\right)-q_\lambda(x,y)^2\\
&=&
\sum_{k=1}^m\lambda^{k-1}\sigma_k+\sum_{i=0}^n\frac{f_i}{\lambda-a_i}+
\sum_{i=0}^n\frac{\mu_i^2}{(\lambda-a_i)^2}.
\end{eqnarray*}

By taking  $V^+=\sigma/2V^{(2)}=\sigma/2\langle
Ax,x\rangle-\sigma/2\langle x,x\rangle^2$ and assuming $\mu_0=0$,
in the limit $a_0 \rightarrow 0$ , we get the Lax representation
for a natural mechanical system in $\R^n$ under the influence of
the the Garnier potential $V=\sigma/2\langle a\mathbf x,\mathbf
x\rangle-\sigma/2\langle\mathbf x,\mathbf x\rangle^2$ obtained by
Antonowicz and Rauch-Wojciechowski \cite{AW} (see also \cite{Su}).
Note that, similarly as in Eilbecktt, Enol'skii, Kuznetsov and
Tsiganov \cite{EEKT}, one can consider the problem within a
framework of $r$-matrix method.

Finally note that the polynomials $V^{(k)}$, as well as the Lax
representation, are well defined for a symmetric ellipsoid
\eqref{ellipsoid}, \eqref{symmetric}. Therefore, repeating the
construction presented in Section 5, we obtain.

\begin{cor}
The system \eqref{H-sep-pot} on a symmetric ellipsoid
\eqref{ellipsoid}, \eqref{symmetric} is completely integrable in a
non-commutative sense by means of integrals \eqref{Ros} and
$$
\tilde f_s=\sum_{i\in I_s} \left(y_i^2+\sum_{k=1}^m \sigma_k
F^{(k)}_i +\frac{\mu_i^2}{x_i^2}+\sum_{j\notin
I_s}\frac{P_{ij}}{a_i-a_j}\right), \quad s=0,1,\dots,r.
$$
Generic trajectories take place over $r+\rho$-dimensional
invariant isotropic tori, spanned by the Hamiltonian vector fields
$X_{\tilde f_s}, X_{P_s}$. Also, the problem is Liouville
integrable by means of integrals $\tilde f_s$ and \eqref{L-JR}.
\end{cor}

\subsection*{Acknowledgments}  This research was supported by the
Serbian Ministry of Science Project 174020, Geometry and Topology
of Manifolds, Classical Mechanics and Integrable Dynamical
Systems. A draft of Section 6 is obtained during author's visiting
UPC, Barcelona in September 2011. Author would like to thanks Yuri
Fedorov for useful remarks and kind hospitality.

\end{document}